\documentclass{amsart}
\usepackage{amsmath,amsfonts,amssymb,amsthm,mathtools}
\usepackage{amscd}
\usepackage{bbm}
\usepackage{enumerate}
\usepackage{galois}
\usepackage{mathrsfs}
\usepackage{xypic}
\usepackage{geometry}
\usepackage{hyperref}

\theoremstyle{plain}
\newtheorem{Th}{Theorem}[section]

\newtheorem{Lem}[Th]{Lemma}

\newtheorem{Ex}{Example}[section]
\newtheorem{Rem}{Remark}[section]

\numberwithin{equation}{section}

\DeclareMathOperator{\ima}{Im}

\newcommand{\derx}{\mathrm{Der}^{\qp_x}(\hm A)}
\newcommand{\derd}{\mathrm{Der}^D(\hm A)}
\newcommand{\diff}[2]{\frac{\partial #1}{\partial #2}}

\newcommand{\qa}{\alpha}
\newcommand{\qb}{\beta}
\newcommand{\qd}{\delta}
\newcommand{\qg}{\gamma}
\newcommand{\qs}{\sigma}
\newcommand{\qt}{\tau}
\newcommand{\qth}{\theta}
\newcommand{\qe}{\varepsilon}

\newcommand{\qp}{\partial}

\newcommand{\Qg}{\Gamma}
\newcommand{\ql}{\lambda}
\newcommand{\Ql}{\Lambda}

\newcommand{\vard}[2]{\frac{\delta #1}{\delta #2}}

\newcommand{\hm}[1]{\hat{\mathcal #1}}

\newcommand{\fk}[2]{\left[#1, #2\right]}

\begin{document}

\title[Linear Reciprocal Transformations and Bihamiltonian Structures]{Variational Bihamiltonian Cohomologies and Integrable Hierarchies III: Linear Reciprocal Transformations}

\author{Si-Qi Liu, Zhe Wang, Youjin Zhang}

\begin{abstract}
For an integrable hierarchy which possesses a bihamiltonian structure with semisimple hydrodynamic limit, we prove that the linear reciprocal transformation with respect to any of its symmetry transforms it to another bihamiltonian integrable hierarchy. Moreover, we show that the central invariants of the bihamiltonian structure  are preserved under such a linear reciprocal transformation.
\end{abstract}

\date{\today}

\maketitle
\tableofcontents


\section{Introduction}\label{intro}

Reciprocal transformations are transformations of independent variables of evolutionary PDEs, which originally appear in the study of gas dynamics, see \cite{rogers1982backlund} and references therein. Such transformations also play important roles in the study of the theory of integrable hierarchies. For example, they are used to  establish relations between different integrable hierarchies and reveal geometric structures behind them. 

Consider a system of evolutionary PDEs of the form
\begin{equation}
\label{AB}
\diff{u^\qa}{t} = A^\qa_\qb(u) u^{\qb}_x+\qe \left(B^\qa_\qb(u) u^\qb_{xx}+C^\qa_{\qb\qg}(u)u^{\qb}_xu^\qg_x\right)+\dots,
\end{equation}
here and henceforth we assume summations over repeated upper and lower Greek indices.
If it has a nontrivial conservation law
\begin{align*}
\diff{a(u,u_x,u_{xx},\dots)}{t} = \diff{b(u,u_x,u_{xx},\dots)}{x},
\end{align*}
then we can define a reciprocal transformation $(x,t)\mapsto (y,s)$ as follows:
\[
dy = a(u)dx+b(u)dt,\quad s = t,
\]
and we call such a reciprocal transformation a nonlinear one. We can also define a linear reciprocal transformation by exchanging the spatial and time variables, i.e., the transformation $(x,t)\mapsto(y,s)$ given by
\[
y = t,\quad s = x.
\]
We can compose these two kinds of reciprocal transformations to obtain general ones. Typical examples of nonlinear reciprocal transformations include the one which transforms the Camassa-Holm equation to the first negative flow of the KdV hierarchy, and the one which transforms the 2-component Camassa-Holm equation to the first negative flow of the AKNS hierarchy \cite{camassa1993integrable,camassa1994new,chen2006two,fuchssteiner1996some,fuchssteiner1981symplectic}.
An example of linear reciprocal transformation is given by the transformation between the extended Toda hierarchy and the nonlinear Schr\"{o}dinger hierarchy \cite{carlet2004extended}. Linear reciprocal transformations are also used in \cite{dubrovin2001normal} to establish relations between a certain type of bihamiltonian integrable hierarchies and Frobenius manifolds.

An important problem concerned with reciprocal transformations is whether such transformations preserve the Hamiltonian property of Hamiltonian evolutionary PDEs. This problem is studied and an affirmative answer is obtained in \cite{ferapontov2003reciprocal,pavlov1995cons,tsarev1991geometry,xue2006bihamiltonian} 
when the system \eqref{AB} is of hydrodynamic type, i.e., of the form
\[
\diff{u^\qa}{t} = A^\qa_\qb(u)u^\qb_x.
\]
It is shown in \cite{ferapontov2003reciprocal} that after performing a nonlinear reciprocal transformation, a Hamiltonian system of hydrodynamic type is transformed to a Hamiltonian system but with a non-local Hamiltonian structure of a certain special type \cite{fera1995conformal}. These non-local Hamiltonian structures are generalized to Jacobi structures \cite{liu2011jacobi} which are certain non-local Hamiltonian structures for deformations of non-local Hamiltonian systems of hydrodynamic type, and their transformation rules under nonlinear reciprocal transformations are given in \cite{liu2011jacobi}. More general non-local Hamiltonian structures are introduced and studied in \cite{lorenzoni2020weakly} and the problem of their classification under reciprocal transformations is studied in \cite{lorenzoni2023miura}.

Unlike nonlinear ones, linear reciprocal transformations preserve the locality of Hamiltonian structures of hydrodynamic type, as it is shown in  \cite{pavlov1995cons,tsarev1991geometry}. Furthermore, if the original system admits a semisimple bihamiltonian structure of hydrodynamic type, then the transformed system again possesses a semisimple bihamiltonian structure of hydrodynamic type \cite{xue2006bihamiltonian}. However, it is not easy to generalize this result to the bihamiltonian dispersionful systems with hydrodynamic limits. 

In this paper we study the transformation rule of the bihamiltonian structure of a bihamiltonian integrable hierarchy which possesses hydrodynamic limit under 
linear reciprocal transformations.
As an illustrative example, we recall here the linear reciprocal transformation which relates the Toda hierarchy to the nonlinear Schr\"{o}dinger hierarchy, as it is given in \cite{carlet2004extended}. The Toda  hierarchy can be described by the following Lax pairs:
\begin{equation}
\label{DB}
\qe\diff{L}{t_k} = [L^k_+,L],\quad k\geq 0,
\end{equation}
with the Lax operator
\[L = \Ql+v+e^u\Ql^{-1},\]
here $\Ql = \exp(\qe \qp_x)$ is the shift operator along the spatial variable. For example, we have
\[
\diff{v}{t_1} = \frac1\qe(\Ql-1)e^{u(x)},\quad \diff{u}{t_1} = \frac1\qe(1-\Ql^{-1})v(x).\]
Introduce new dependent variables 
\[
\varphi(x) = \Ql^{-1}v(x),\quad \rho(x) = e^{u(x)},
\]
then we have
\begin{equation}
\label{DC}
\qe\diff{\varphi}{t_1} = \rho-\rho^-,\quad \qe\diff{\rho}{t_1} = \rho(\varphi^+-\varphi).
\end{equation}
Here and henceforth we use the notation
\[
\varphi^+:=\Ql \varphi,\quad \rho^+:=\Ql\rho,\quad \varphi^-:=\Ql^{-1} \varphi,\quad \rho^-:=\Ql^{-1}\rho.
\]
Similarly, we have the following second flow of the Toda hierarchy \eqref{DB}:
\begin{align}
\label{DD}
\qe\diff{\varphi}{t_2} &=\varphi^+\rho+\varphi\rho-\varphi\rho^--\varphi^-\rho^-\\
\label{DE}
 \qe\diff{\rho}{t_2} &=\rho\varphi^+\varphi^+-\rho\varphi^2+\rho\rho^+-\rho\rho^-.
\end{align}
Let us perform the linear reciprocal transformation with respect to the flow \eqref{DC}, i.e., we represent $\diff{}{x}$ in terms of the derivatives with respect to the flow $\diff{}{t_1}$:
\begin{align*}
\diff{\varphi}{x} &= \frac1\rho\diff\rho{t_1}+\qe\left(\frac{1}{2\rho^2}\diff\rho{t_1}\diff\varphi{t_1}-\frac{1}{2\rho}\frac{\qp^2\varphi}{\qp t_1^2}\right)+O(\qe^2),\\
\diff{\rho}{x}&=\diff\varphi{t_1}+\qe\left(-\frac{1}{2\rho^2}\left(\diff\rho{t_1}\right)^2+\frac{1}{2\rho}\frac{\qp^2\rho}{\qp t_1^2}\right)+O(\qe^2).
\end{align*}
This linear reciprocal transformation can be applied to other flows of the Toda  hierarchy. For example, a straightforward computation yields
\[
\diff{\varphi}{t_2} = 2\varphi\diff{\varphi}{t_1}+2\diff{\rho}{t_1}-\qe \frac{\qp^2\varphi}{\qp t_1^2},\quad \diff{\rho}{t_2} = 2\varphi\diff{\rho}{t_1}+2\rho\diff{\varphi}{t_1}+\qe \frac{\qp^2\rho}{\qp t_1^2}.
\]
If we rewrite $t_1$ as the new spatial variable $X$, the above flow coincides with the nonlinear Schr\"{o}dinger equation:
\[
\diff{\varphi}{t_2} = 2\varphi\varphi_X+2\rho_X-\qe \varphi_{XX},\quad \diff{\rho}{t_2} = 2\varphi\rho_X+2\rho\varphi_X+\qe \rho_{XX}.
\]  Both  the Toda hierarchy and the nonlinear Schr\"{o}dinger hierarchy are bihamiltonian. The bihamiltonian structure of the Toda hierarchy can be represented by the following Hamiltonian operators in terms of the variables $(\varphi,\rho)$ by
\[\mathcal P_0 = \frac1\qe\begin{pmatrix}0&\rho-\rho^-\Ql^{-1}\\[5pt] -\rho+\rho\Ql&0\end{pmatrix},\quad \mathcal P_1 = \frac1\qe\begin{pmatrix}\rho\Ql-\rho^-\Ql^{-1}&\varphi\rho-\varphi\rho^-\Ql^{-1}\\[5pt] -\varphi\rho+\varphi^+\rho\Ql&\rho\rho^+\Ql-\rho\rho^-\Ql^{-1}\end{pmatrix}.\]
And the Hamiltonian operators of bihamiltonian structure of the nonlinear Schr\"{o}dinger equation in terms of $(\varphi,\rho)$ reads
\[\mathcal P_0 = \begin{pmatrix}0&\qp_X\\[5pt]\qp_X&0\end{pmatrix},\quad \mathcal P_1 = \begin{pmatrix}2\qp_X&\varphi\qp_X+\varphi_X-\qe\qp_X^2\\[5pt] \varphi\qp_X+\qe\qp_X^2&2\rho\qp_X+\rho_X\end{pmatrix}.\]

For a general bihamiltonian integrable hierarchy we prove, under the condition of semisimplicity, that linear reciprocal transformations always preserve the bihamiltonian property. To state this result in a more precise way, let us consider the following bihamiltonian evolutionary PDEs
\begin{equation}
\label{DM}
\diff{u^\qa}{t} = \mathcal P^{\qa\qb}_0\vard{H_0}{u^\qb} = \mathcal P^{\qa\qb}_1\vard{H_1}{u^\qb} = A^\qa_\qb u^{\qb}_x+O(\qe)
\end{equation}
with Hamiltonian operators $\mathcal P_i = (\mathcal P^{\qa\qb}_i)$ and Hamiltonians $H_i$ of the form
\begin{align*}
\mathcal P^{\qa\qb}_i &=g^{\qa\qb}_i\qp_x+\Qg^{\qa\qb}_{\qg}u^{\qg}_x+\qe(A^{\qa\qb}_i\qp_x^2+B^{\qa\qb}_{i;\qg}u^{\qg}_x\qp_x+C^{\qa\qb}_{i;\qg\ql}u^{\qg}_xu^{\ql}_x)+O(\qe^2),\\
H_i &= \int h_i(u)+\qe a_{i;\qa}(u)u^\qa_x+O(\qe^2),
\end{align*}
where $\qe$ is the dispersion parameter. We require that the leading terms of the above bihamiltonian structure are semisimple, i.e., the roots of characteristic equation \[\det(g_1^{\qa\qb}-\ql g_0^{\qa\qb})=0\] are distinct and non-constant. By exchanging the spatial and time variables of the system \eqref{DM}, we arrive at
\[
\diff{u^\qa}{x} = (A^{-1})^\qa_\qb \diff{u^\qb}{t}+O(\qe).
\] 
If we still use $x$ to denote the spatial variable and use $\tilde t$ to denote the time variable of the above system, we obtain from \eqref{DM} by the linear reciprocal transformation the following system of evolutionary PDEs:
\begin{equation}
\label{DP}
\diff{u^\qa}{\tilde t} = (A^{-1})^\qa_\qb u^\qb_x+O(\qe).
\end{equation}

Let us state the main result of the present paper.
\begin{Th}
\label{AE}
The system of evolutionary PDEs \eqref{DP} has the following properties:
\begin{enumerate}
\item It possesses a bihamiltonian structure with semisimple  hydrodymanic limit. The central invariants of this bihamiltonian structure coincide with those of the bihamiltonian structure of \eqref{DM}.
\item After performing the linear reciprocal transformation, any symmetry of the system \eqref{DM} of the form
\[
\diff{u^\qa}{s} = \mathcal P^{\qa\qb}_0\vard{G_0}{u^\qb} = \mathcal P^{\qa\qb}_1\vard{G_1}{u^\qb}
\]
is transformed to a symmetry of \eqref{DP} which shares the same bihamiltonian structure with \eqref{DP}.
\end{enumerate}
\end{Th}

This result provides us a convenient tool to classify bihamiltonian integrable hierarchies with semisimple hydrodynamic limits under linear reciprocal transformations and Miura-type transformations. More precisely, given two bihamiltonian integrable hierarchies with hydrodynamic limits, we obtain a very simple method to determine whether we can transform one hierarchy to the other by performing a linear reciprocal transformation composed by a Miura-type transformation. For example, we show that the equivalence classes of scalar bihamiltonian integrable hierarchy under linear reciprocal transformations and Miura-type transformations are parametrized by smooth functions of a single variable.

This paper is organized as follows. In Sect.\,\ref{vbh}, we review the theory of variational bihamiltonian cohomology which is our main tool to prove that linear reciprocal transformations preserve bihamiltonian property. In Sect.\,\ref{pro} we give the proof of Theorem \ref{AE} and use it to study the problem of classification of bihamiltonian integrable hierarchies under linear reciprocal transformations and Miura-type transformations. In Sect.\ref{ex} we present some examples, and prove some general results on relations of linear reciprocal transformations with the theory of Frobenius manifolds and  Dubrovin-Zhang hierarchies. Finally in Sect.\ref{conc} we give some concluding remarks.

\section{Variational bihamiltonian cohomology}\label{vbh}
This section is a technical preparation for proving the main result of the present paper.
\subsection{Infinite jet spaces and bihamiltonian structures}
We start by recalling the representation of bihamiltonian structures in terms of local functionals on infinite jet spaces of super manifolds, which is also called the $\qth$-formalism in the literature \cite{carlet2018deformations,getzler2002darboux,liu2011jacobi,liu2013bihamiltonian}. One can refer to \cite{liu2018lecture} for a detailed introduction.

Let $M$ be a smooth manifold of dimension $n$ and let $\hat M = \Pi T^*M$ be the super manifold of dimension $(n|n)$ obtained by reversing the parity of the fiber of $T^*M$. Denote by $(u^1,\dots,u^n;\qth_1,\dots,\qth_n)$ a system of local canonical coordinates on $T^*M$, then the manifold $\hat M$ can be described locally by the same coordinates with the parity of the fiber coordinates $\qth_\qa$ reversed, i.e., they satisfy the anti-commuting relations
\[
\qth_\qa\qth_\qb+\qth_\qb\qth_\qa=0.
\]
These fiber coordinates are hence called odd variables. We denote by $J^\infty(\hat M)$ the infinite jet space of $\hat M$ which can be described by local coordinates $(u^{\qa,s};\qth_\qa^s)$, where $\qa  =1,\dots,n$, $s\geq 0$ and $u^{\qa,0} = u^\qa$, $\qth_\qa^0 = \qth_\qa$. Let us denote by $\hm A$ the space of differential polynomials on $J^\infty(\hat M)$, locally it is the ring
\[
C^\infty(u)[[u^{\qa,s+1},\qth_\qa^s\mid \qa  =1,\dots,n; s\geq 0 ]].
\]
There exists a global vector field $\qp_x$ on $J^\infty(\hat M)$ defined by
\[
\qp_x = \sum_{s\geq 0}u^{\qa,s+1}\diff{}{u^{\qa,s}}+\qth_\qa^{s+1}\diff{}{\qth_\qa^s}.
\]
By using this vector field, we define the space $\hm F$ of local functionals as \[\hm F:=\hm A/\qp_x\hm A,\]and for $f\in\hm A$, we denote by $\int f$ the image of $f$ in $\hm F$. Furthermore, we define two gradations on $\hm A$ and $\hm F$ by setting 
\[
\deg_x u^{\qa,s} = s,\quad \deg_x \qth_\qa^s = s,\quad \deg_\qth u^{\qa,s} = 0,\quad \deg_\qth \qth_\qa^s = 1.
\]
We denote the homogeneous subspaces with respect to these gradations by
\[
\hm A^p = \{f\in\hm A\mid\deg_\qth f = p\},\quad \hm A_d = \{f\in\hm A\mid\deg_x f = d\},\quad \hm A^p_d = \hm A^p\cap\hm A_d.
\]
Since the vector field $\qp_x$ is homogeneous with respect to both $\deg_x$ and $\deg_\qth$, the quotient space $\hm F$ has natural gradations induced from those on $\hm A$. We use $\hm F^p$, $\hm F_d$ and $\hm F^p_d$ to denote the corresponding subspaces of homogeneous elements.

There is a graded Lie algebra structure on $\hm F$ given by the so-called Schouten-Nijenhuis bracket. Define the bilinear bracket $[-,-]:\hm F\times\hm F\to\hm F$ by
\begin{equation}\label{DN}
[P,Q] = \int\vard{P}{\qth_\qa}\vard{Q}{u^\qa}+(-1)^p\vard{P}{u^\qa}\vard{Q}{\qth_\qa},\quad P\in\hm F^p,\quad Q\in\hm F,
\end{equation}
here the variational derivative is defined for any $Q\in\hm F$ by
\[
\vard{Q}{u^\qa} = \sum_{s\geq 0}(-\qp_x)^s\diff{f}{u^{\qa,s}},\quad \vard{Q}{\qth_\qa} = \sum_{s\geq 0}(-\qp_x)^s\diff{f}{\qth_\qa^s},\quad Q = \int f.
\]
This bracket is antisymmetric and satisfies the graded Jacobi identity. More precisely, for $P\in\hm F^p$, $Q\in\hm F^q$ and $R\in\hm F^r$, we have
\begin{align}
&[P,Q] = (-1)^{pq}[Q,P],\\
\label{AI}
&(-1)^{pr}[[P,Q],R]+(-1)^{qp}[[Q,R],P]+(-1)^{rq}[[R,P],Q] = 0.
\end{align}

The Schouten-Nijenhuis bracket is preserved under Miura-type transformations. A Miura-type transofrmation is a change of coordinate on $J^\infty(M)$ that preserves the action of $\qp_x$, i.e., it is of the form
\begin{equation}
\label{AF}
u^{\qa,s}\mapsto w^{\qa,s} = \sum_{k\geq 0} \qp_x^sf^\qa_k(u),\quad f^\qa_k(u)\in\hm A^0_k,\quad s\geq 0,
\end{equation}
where the leading terms satisfy the condition
\[
\det\left(\diff{f^\qa_0(u)}{u^\qb}\right)\neq 0.
\]
Given a Miura-type transformation of the form \eqref{AF}, it is proved in \cite{liu2011jacobi} that the following transformation rules for odd variables preserve the Schouten-Nijenhuis bracket:
\[
\qs_\qa^s = \qp_x^s\sum_{k\geq 0}(-\qp_x)^k\left(\diff{u^{\qb}}{w^{\qa,k}}\qth_\qb\right),\quad s\geq 0,
\]
i.e., if we represent two local functionals $P, Q\in\hm F$ in terms of the new coordinates $w^{\qa,s}$ and $\qs_\qa^s$, then their Schouten-Nijenhuis bracket can be computed by using the formula \eqref{DN} with $(u^\qa;\qth_\qa)$ replaced by $(w^\qa;\qs_\qa)$.

We define a Hamlitonian structure to be a local functional $P\in\hm F^2$ such that $[P,P]$ = 0 and a bihamiltonian structure $(P_0,P_1)$ to be a pair of Hamiltonian structures such that $[P_0,P_1] = 0$. We also define a system of evolutionary PDEs to be a local functional $X\in\hm F^1$, and we call it is Hamiltonian if there exists a Hamiltonian structure $P$ and a local functional $H\in\hm F^0$ such that $X = -[P,H]$. Similarly we call $X$ is bihamiltonian if there exists a bihamiltonian structure $(P_0,P_1)$ and two local functionals $H_0,H_1\in\hm F^0$ such that $X = -[P_0,H_0] = -[P_1,H_1]$.

In the literature, a system of Hamiltonian evolutionary PDEs is usually represented as follows:
\[
\diff{u^\qa}{t} = X^\qa =  \mathcal P^{\qa\qb}\vard{H}{u^\qb},\quad X^\qa\in\hm A^0,\quad H\in\hm F^0,
\]
where $\mathcal P = (\mathcal P^{\qa\qb})$ is the Hamiltonian operator which has the form
\[
\mathcal P^{\qa\qb} = \sum_{s\geq 0}P^{\qa\qb}_s\qp_x^s,\quad P^{\qa\qb}_s\in\hm A^0.
\]
Then in the language of local functionals, the above system of evolutionary PDEs corresponds to the local functional \[
X = \int X^\qa\qth_\qa\in\hm F^1
\] 
and the Hamiltonian operator $\mathcal P$ corresponds to a local functional
\[
P = \frac{1}{2}\int \sum_{s\geq 0}P^{\qa\qb}_s\qth_\qa\qth_\qb^s\in\hm F^2
\]
which satisfies the condition $[P,P]=0$. These lcoal functionals are related by the formula
\[
X = -[P,H].
\]
Conversely, given local functionals $H\in \hm F^0$, $X\in\hm F^1$ and $P\in\hm F^2$ which satisfy the conditions $X = -[H,P]$ and $[P,P] = 0$, then we can define a Hamiltonian operator $\mathcal P = (\mathcal P^{\qa\qb})$ with
\begin{equation}
\label{AG}
\mathcal P^{\qa\qb} = \sum_{s\geq 0}\diff{}{\qth_\qb^s}\left(\vard{P}{\qth_\qa}\right)\qp_x^s,
\end{equation}
and a system of evolutionary PDEs
\[
\diff{u^\qa}{t} = \vard{X}{\qth_\qa},
\]
which is a Hamiltonian system in the sense that
\[
\diff{u^\qa}{t} = \mathcal P^{\qa\qb}\vard{H}{u^\qb}.
\]

In the remaining part of the subsection, we review the deformation theory of (bi)Hamiltonian structures. We start by recalling the concept of Hamiltonian structures of hydrodynamic type. For a Hamiltonian structure $P\in\hm F^2_1$, the corresponding Hamiltonian operator $\mathcal P$ defined in \eqref{AG} has the form
\[
\mathcal P^{\qa\qb} = g^{\qa\qb}\qp_x+\Qg^{\qa\qb}_\qg u^{\qg,1}.
\]
If $\det(g^{\qa\qb})\neq 0$, then we call $P$ is of hydrodynamic type. It is proved in \cite{dubrovin1996hamiltonian}
that in this case $(g_{\qa\qb}):=(g^{\qa\qb})^{-1}$ defines a flat metric on $M$ and $\Qg^\qa_{\qb\qg} =-g_{\qb\mu}\Qg^{\mu \qa}_\qg$ are the Christoffel symbols of the Levi-Civita connection of $(g_{\qa\qb})$.

Given a bihamiltonian structure $(P_0,P_1)$ of hydrodynamic type, then we have a pair of flat (contravariant) metrics $(g_0^{\qa\qb},g_1^{\qa\qb})$ which forms a flat pencil \cite{dubrovin1996geometry}. This bihamiltonian structure is called semisimple if the roots of the characteristic equation
\[
\det(g_1^{\qa\qb}-\ql g_0^{\qa\qb}) = 0
\]
are all distinct and non-constant. In this case, if we denote by $\ql^1,\dots,\ql^n$ the $n$ distinct solutions of the above equation, then it is proved in \cite{ferapontov2001compatible} that they can be used as local coordinates of $M$, and they are called the canonical coordinates of the semisimple bihamiltonian structure $(P_0,P_1)$. In terms of the canonical coordinates, the two metrics $g_0^{ij},g_1^{ij}$ are diagonal and have the form
\[
g_0^{ij}(\ql) = f^i(\ql)\qd^{ij},\quad g_1^{ij}(\ql) = \ql^if^i(\ql)\qd^{ij},
\]
here $f^i(\ql)$ are some non-zero functions.

An important problem in the theory of integrable hierarchy is to classify the Hamiltonian (bihamiltonian) structures with hydrodynamic limits, or equivalently, to study the deformation theory of hydrodynamic Hamiltonian (bihamiltonian) structures. The classification is completed in a series of papers \cite{carlet2018central,carlet2018deformations,degiovanni2005deformation,DLZ-1,dubrovin2001normal,getzler2002darboux,liu2013bihamiltonian} by computing the so called Hamiltonian (bihamiltonian) cohomology. The deformation problem can be described as follows. Given a Hamiltonian structure $P$ of hydrodynamic type, we want to classify all the possible deformations
\[
\tilde P = P+Q,\quad Q\in\hm F^2_{\geq 2}
\]
such that $\tilde P$ is a Hamiltonian structure. We regard two deformations to be equivalent if one can be obtained from the other by a Miura-type transformation. Similarly we study the deformation problem of a semisimple bihamiltonian structure $(P_0,P_1)$ of hydrodynamic type. The classification results of Hamiltonian (bihamiltonian) structures are summarized in the following theorem.
\begin{Th}[\cite{carlet2018deformations,degiovanni2005deformation,DLZ-1,getzler2002darboux}]
Let $P$ be a Hamiltonian structure of hydrodynamic type and $(P_0,P_1)$ be a semisimple bihamiltonian structure of hydrodynamic type. 
\begin{enumerate}
\item All the deformations of $P$ are trivial, i.e.,they are Miura-equivalent to the trivial deformation $\tilde P = P$.
\item Equivalent classes of deformations of $(P_0,P_1)$ are parmeterized by $n$ functions of a single variable. These functions are called the central invariants of a deformation.
\end{enumerate}
\end{Th}

Let us give a more detailed description on the central invairants of any deformation $(\tilde P_0,\tilde P_1)$ of a semisimple bihamiltonian structure $(P_0,P_1)$ of hydrodynamic type. In terms of the canonical coordinates $(\ql^1,\dots,\ql^n)$ of $(P_0,P_1)$, the Hamiltonian operators $\tilde{\mathcal P}_0$ and $\tilde{\mathcal P}_1$ of $\tilde P_0$ and $\tilde P_1$ can be represented in the form
\[
\tilde{\mathcal P}_a^{ij} = \sum_{k\geq 1}\sum_{s=0}^k \tilde P^{ij}_{a;k,s}\qp_x^s,\quad  \tilde P^{ij}_{a;k,s}\in\hm A^0_{k-s},\quad a = 0,1,
\]
in particular we know that
\[
\tilde P^{ij}_{0;1,1} = f^i(\ql)\qd^{ij},\quad \tilde P^{ij}_{1;1,1} = \ql^if^i(\ql)\qd^{ij}.
\]
Then the central invariants can be explicitly computed by the following formulae:
\begin{equation}
\label{BA}
c_i(\ql^i) = \frac{1}{3(f^i)^2}\left(\tilde P^{ii}_{1;3,3}-\ql^i\tilde P^{ii}_{0;3,3}+\sum_{j\neq i}\frac{\left(\tilde P^{ji}_{1;2,2}-\ql^i\tilde P^{ji}_{0;2,2}\right)^2}{f^j(\ql^j-\ql^i)}\right),\quad i = 1,\dots,n.
\end{equation}
Note that $c_i(\ql^i)$ is a function depending only on the variable $\ql^i$. Conversely, given $n$ functions $c_1(\ql^1),\dots,c_n(\ql^n)$ of a single variable, we can construct a unique equivalence class of  deformations of $(P_0,P_1)$ such that its central invariants are exactly given by $c_1(\ql^1),\dots,c_n(\ql^n)$. Such a deformation, up to Miura-type transformations, can be represented by
\[
\tilde P_0 = P_0,\quad \tilde P_1 = P_1+[P_0,Q]+R,\quad R\in\hm F^2_{\geq 4},
\]
where $Q\in\hm F^1_2$ is given by 
\begin{equation}
\label{AL}
Q = \left[P_1,\int\sum_i c_i(\ql^i)\ql^{i,1}\log\ql^{i,1}\right]-\left[P_0,\int\sum_i \ql^ic_i(\ql^i)\ql^{i,1}\log\ql^{i,1}\right].
\end{equation}

The following theorem says that any deformation of a semisimple bihamiltonian structure induces a unique deformation of the corresponding bihamiltonian hierarchy of hydrodynamic type.
\begin{Th}[\cite{DLZ-1}]
\label{BX}Let $(P_0,P_1)$ be a semisimple bihamiltonian structure of hydrodynamic type and let $X\in\hm F^1_1$ satisfy the condition $[X,P_0] = [X,P_1]=0$. If $(\tilde P_0,\tilde P_1)$ is any deformation of $(P_0,P_1)$, then there exists a unique deformation $\tilde X$ of $X$ such that $[\tilde X,\tilde P_0] = [\tilde X,\tilde P_1]=0$.
\end{Th}

\subsection{Variational bihamiltonian cohomology}

In the previous subsection, we mainly consider the space of local functionals $\hm F$. In this subsection, we focus on the space of derivations $\derx$ and it is defined by
\[
\derx = \bigoplus_{p\geq-1}\derx^p,
\]
where $\derx^p$ is the space of derivations with $\deg_\qth = p$ that commute with $\qp_x$. In another word, an element $X\in\derx^p$ is a linear map $\hm A\to\hm A$ such that
\[
X(\hm A^q)\subseteq \hm A^{p+q},\quad [X,\qp_x] = 0,
\] 
and for any differential polynomials $f\in\hm A^q$ and $g\in\hm A$ we have
\[
X(fg) = X(f)g+(-1)^{pq}fX(g).
\]
The differential degree $\deg_x$ of $\hm A$  induces another gradation on $\derx$, in terms of which a homogeneous element $X$ of degree $d$ is specified by the condition $X(\hm A_t)\subseteq\hm A_{d+t}$. We denote by $\derx_d$ the space of derivations $X$ with differential degree $d$ and by $\derx^p_d$ the space $\derx^p\cap\derx_d$.
\begin{Rem}
In what follows we will call an element of $\derx$ a vector field or a flow. In particular, we call an element of $\derx^{2n-1}$ for $n\geq 0$ an odd flow.
\end{Rem}

Given a local functional $X\in\hm F^p$, we can define a derivation $D_X\in\derx^{p-1}$ as follows:
\begin{equation}
\label{AH}
D_X = \sum_{s\geq 0}\qp_x^s\left(\vard{X}{\qth_\qa}\right)\diff{}{u^{\qa,s}}+(-1)^p\qp_x^s\left(\vard{X}{u^{\qa}}\right)\diff{}{\qth_\qa^s}.
\end{equation}
Such a derivation is actually the Hamiltonian vector field of the local functional $X$ with respect to the canonical symplectic structure on $\hat M$. The definition \eqref{AH} yields a map 
\begin{equation}\label{DO}\hm F\to\derx:\quad X\mapsto D_X\end{equation}
and we denote by $\derd$ the image of this map. We call elements in $\derd$ as \emph{ derivations of $D$-type}.

It is proved in \cite{liu2018lecture} that the map \eqref{DO} is a homomorphism of graded Lie algebra, where the graded Lie algebra structure on $\derx$ is given by the natural graded commutator of derivations. More precisely, for $P\in\hm F^p$ and $Q\in\hm F$, we have
\begin{equation}
\label{AJ}
D_{[P,Q]} = (-1)^{p-1}[D_P,D_Q].
\end{equation}
This relation allows us to make super extensions for any evolutionary Hamiltonian PDEs. Given the following system of evolutionary Hamiltonian PDEs
\[
\diff{u^\qa}{t} = X^\qa = \mathcal P^{\qa\qb}\vard{H}{u^\qb},
\]
we have seen in the previous subsection that it corresponds to local functionals $H\in \hm F^0$, $X\in\hm F^1$ and $P\in\hm F^2$ such that $X = -[P,H]$. Then it follows from the graded Jacobi identity \eqref{AI} and the condition $[P,P]= 0$ that $[X,P] = 0$, hence $[D_X,D_P] = 0$. Therefore  the above Hamiltonian system admits the following compatible super extension
\begin{align*}
\diff{u^\qa}{t} &= D_X u^\qa,\quad \diff{\qth_\qa}{t} = D_X\qth_\qa,\\
\diff{u^\qa}{\qt} &= D_P u^\qa,\quad \diff{\qth_\qa}{\qt} = D_P\qth_\qa.
\end{align*}
Conversely if we have two derivations $X\in\derx^0$ and $P\in\derx^1$ such that $[X,P] = 0$ and $[P,P] = 0$, then if both $X$ and $P$ are $D$-type, we conclude that the following system of evolutionary PDEs
\[
\diff{u^\qa}{t} = X(u^\qa)
\]
is a Hamiltonian system.

\begin{Rem}
When $P\in\hm F^2$ is a Hamiltonian structure with hydrodynamic leading term, then for $X\in\hm F^1$, the existence of $H\in\hm F^0$ such that $X = -[H,P]$ is equivalent to the condition that $[X,P] = 0$.
\end{Rem}

Given a Hamiltonian structure $P\in\hm F^2$, it follows from the equation \eqref{AJ} that $[D_P,D_P] = 0$. Let us denote by $d_P$ the adjoint operator $[D_P,-]:\derx\to \derx$, then we see that $d_P^2 = 0$ and hence we have a complex $(\derx,d_P)$. Similarly, for a bihamiltonian structure $(P_0,P_1)$, we have a bi-complex $(\derx,d_{P_0},d_{P_1})$. We can define the cohomology groups for these complexes as follows:
\begin{align*}
H^p_d(\derx, P)&:=\frac{\derx^p_d\cap\ker d_P}{\derx^p_d\cap\ima d_P},\\[3pt]
VBH^p_d(\derx, P_0,P_1)&:=\frac{\derx^p_d\cap\ker d_{P_0}\cap\ker d_{P_1}}{\derx^p_d\cap\ima d_{P_0}d_{P_1}}.
\end{align*}
These cohomology groups are called variational Hamiltonian (bihamiltonian) cohomology groups and are introduced and studied in the paper \cite{liu2023variational}, also see the relevant sections in \cite{liu2022variational}. They are originally introduced to study Virasoro symmetries of bihamiltonian integrable hierarchies. In this paper we will use them to study properties of linear reciprocal transformations, to this end we need the following theorem. 
\begin{Th}[\cite{liu2023variational}]
\label{AK}
Let $P$ be a Hamiltonian structure of hydrodynamic type and $(P_0,P_1)$ be a semisimple bihamiltonian structure. Then we have
\begin{enumerate}
\item $H^p_d(\derx,P) = 0$ for any $p\geq -1$ and $d>0$,
\item $VBH^0_{\geq 2}(\derx, P_0,P_1) = 0$ and $VBH^1_{\geq 4}(\derx, P_0,P_1)=0$,
\item  $VBH^1_{3}(\derx, P_0,P_1)\cong\bigoplus_{i=1}^nC^\infty(\mathbb R)$.
\end{enumerate}
\end{Th}

Let us describe the cohomology group $VBH^1_{3}(\derx, P_0,P_1)$ in a more detailed way. Given a derivation $X\in \derx^1_3\cap\ker d_{P_0}\cap\ker d_{P_1}$, there exist $n$ functions $c_i(\ql^i)$ of a single variable and a unique derivation $Y\in\derx^{-1}_1$ such that
\begin{equation}
\label{AM}
X = D_{[P_0,Q]}+d_{P_0}d_{P_1}Y,\end{equation}
where $Q\in\hm F^1_2$ is given by the formula \eqref{AL}. And conversely, for any $n$ functions $c_i(\ql^i)$ of a single variable, the derivation $D_{[P_0,Q]}$ defines a cohomology class in the group $VBH^1_{3}(\derx, P_0,P_1)$.

\section{Proofs of the main result}\label{pro}
In this section we prove Theorem \ref{AE}. To explain the basic ideas of the proof, we first consider the linear reciprocal transformation of the KdV equation
\[
\diff{u}{t} = uu_x+\frac{\qe^2}{12}u_{xxx}.
\]
It has a bihamiltonian structure with the Hamiltonian operators
\begin{equation}
\label{CD}
\mathcal P_0 = \qp_x,\quad \mathcal P_1 = u\qp_x+\frac12u_x+\frac{\qe^2}{8}\qp_x^3.
\end{equation}
Note that in this example, $u$ is already the canonical coordinate of the bihamiltonian structure and the central invariant is $\frac{1}{24}$.

According to the discussion given in Sect.\,\ref{vbh}, by introducing the odd variable $\qth$, the above bihamiltonian system can be represented by the local functionals
\[
X = \int \left(uu_x+\frac{\qe^2}{12}u_{xxx}\right)\qth,\quad P_0 = \frac12\int\qth\qth^1,\quad P_1 = \frac12\int u\qth\qth^1+\frac{\qe^2}{8}\qth\qth^3.
\]
Therefore the KdV equation admits the following compatible super extension
\begin{align}
\label{AN}
\diff{u}{t} &= uu_x+\frac{\qe^2}{12}u_{xxx},& &\diff{\qth}{t}= u\qth^1+\frac{\qe^2}{12}\qth^3,\\
\label{AO}
\diff{u}{\qt_0} &= \qth^1,& &\diff{\qth}{\qt_0} = 0,\\
\label{AP}
\diff{u}{\qt_1}&=u\qth^1+\frac12u_x\qth+\frac{\qe^2}{8}\qth^3,& &\diff{\qth}{\qt_0} = \frac12\qth\qth^1.
\end{align}
We perform the linear reciprocal transformation by exchanging the spatial variable $x$ and the time variable $t$. In another word, we represent $u_x$ and $\qth^1$ by derivatives of $u$ and $\qth$ with respect to $t$. A straightforward computation shows that
\begin{align*}
\diff{u}{x}&=\frac1u\diff{u}{t}+\qe^2\left(-\frac{1}{12u^4}\frac{\qp^3u}{\qp t^3}+\frac{3}{4u^5}\diff{u}{t}\frac{\qp^2u}{\qp t^2}-\frac{1}{u^6}\left(\diff{u}{t}\right)^3\right)+O(\qe^4),\\
\diff{\qth}{x}&=\frac1u\diff{\qth}{t}+\qe^2\left(-\frac{1}{12u^4}\frac{\qp^3\qth}{\qp t^3}+\frac{1}{2u^5}\diff{u}{t}\frac{\qp^2\qth}{\qp t^2}+\frac{1}{4u^5}\diff{\qth}{t}\frac{\qp^2u}{\qp t^2}-\frac{1}{u^6}\left(\diff{u}{t}\right)^2\diff{\qth}{t}\right)+O(\qe^4).
\end{align*}
Let us adopt the usual convention and still use $x$ to denote the spatial variable and $\tilde t$ to denote the time variable, then the linear reciprocal transformation of the equation \eqref{AN} reads
\begin{align}
\label{AQ}
\diff{u}{\tilde t}&=\frac{u_x}{u}+\qe^2\left(-\frac{u_{xxx}}{12u^4}+\frac{3u_xu_{xx}}{4u^5}-\frac{u_x^3}{u^6}\right)+O(\qe^4),\\
\label{AR}
\diff{\qth}{\tilde t}&=\frac{\qth^1}{u}+\qe^2\left(-\frac{\qth^3}{12u^4}+\frac{u_x\qth^2}{2u^5}+\frac{u_{xx}\qth^1}{4u^5}-\frac{u_x^2\qth^1}{u^6}\right)+O(\qe^4).
\end{align}
Now we can apply the reciprocal transformation to the odd flows \eqref{AO} and \eqref{AP} by using the substitution
\[
\qp_x^s u\mapsto \frac{\qp^s u}{\qp \tilde t^s},\quad \qth^s\mapsto \frac{\qp^s \qth}{\qp \tilde t^s},
\]
where $\diff{}{\tilde t}$ is the flow given by the equations \eqref{AQ} and \eqref{AR}, and we obtain
\begin{align}
\label{AS}
\diff{u}{\tilde\qt_0} =&\, \frac{\qth^1}{u}+\qe^2\left(-\frac{\qth^3}{12u^4}+\frac{u_x\qth^2}{2u^5}+\frac{u_{xx}\qth^1}{4u^5}-\frac{u_x^2\qth^1}{u^6}\right)+O(\qe^4),\\
\label{AT}
\diff{\qth}{\tilde\qt_0}=&\,0,\\
\label{AU}
\diff{u}{\tilde \qt_1}=&\,\qth^1+\frac{u_x\qth}{2u}+\qe^2\left(\frac{\qth^3}{24u^3}-\frac{u_x\qth^2}{4u^4}-\frac{u_{xx}\qth^1}{8u^4}+\frac{u_x^2\qth^1}{2u^5}\right.\\
\notag
&\left.-\frac{u_{xxx}\qth}{24u^4}+\frac{3u_xu_{xx}\qth}{8u^5}-\frac{u_x^3\qth}{2u^6}\right)+O(\qe^4),\\
\label{AV}
\diff{\qth}{\tilde\qt_1}=&\,\frac{\qth\qth^1}{2u}+\qe^2\left(-\frac{\qth\qth^3}{24u^4}+\frac{u_x\qth\qth^2}{4u^5}+\frac{u_{xx}\qth\qth^1}{8u^5}-\frac{u_x^2\qth\qth^1}{2u^6}\right)+O(\qe^4).
\end{align}
Thus by performing the linear reciprocal transformation to the system \eqref{AN}--\eqref{AP}, we arrive at another system of evolutionary PDEs \eqref{AQ}--\eqref{AV}. Since the linear reciprocal transformation preserves the commutativity of the flows, the system of evolutionary  PDEs \eqref{AQ}--\eqref{AV} is compatible. Therefore if these flows are of $D$-type, then from the discussion given in Sect.\,\ref{vbh}, it follows that the equation \eqref{AQ} is bihamiltonian.

However, it is easy to see that the flows \eqref{AQ}--\eqref{AV} are NOT of $D$-type. To solve this problem, we introduce a new odd variable $\qs$ by
\[
\qs = u\qth+\qe^2\left(\frac{\qth^2}{4u^2}-\frac{u_x\qth^1}{4u^3}\right)+O(\qe^4),
\]
or equivalently we have
\[
\qth = \frac{\qs}{u}+\qe^2\left(-\frac{\qs^2}{4u^4}+\frac{3u_x\qs^1}{4u^5}+\frac{u_{xx}\qs}{4u^5}-\frac{3u_x^2\qs}{4u^6}\right)+O(\qe^4).
\]
We can rewrite the system \eqref{AQ}--\eqref{AV} in terms of the variables $u$ and $\qs$ as
\begin{align}
\label{AW}
\diff{u}{\tilde t}=&\,\frac{u_x}{u}+\qe^2\left(-\frac{u_{xxx}}{12u^4}+\frac{3u_xu_{xx}}{4u^5}-\frac{u_x^3}{u^6}\right)+O(\qe^4),\\
\label{AX}
\diff{\qs}{\tilde t}=&\,\frac{\qs^1}{u}+\qe^2\left(-\frac{\qs^3}{12u^3}+\frac{u_x\qs^2}{4u^5}+\frac{u_{xx}\qs^1}{4u^5}-\frac{u_x^2\qs^1}{2u^6}\right.\\
\notag
&\left.+\frac{u_xu_{xx}\qs}{4u^6}-\frac{u_x^3\qs}{2u^7}\right)+O(\qe^4),\\
\label{AY}
\diff{u}{\tilde \qt_0}=&\,\frac{\qs^1}{u^2}-\frac{u_x\qs}{u^3}+\qe^2\left(-\frac{\qs^3}{3u^5}+\frac{5u_x\qs^2}{2u^6}+\frac{3u_{xx}\qs^1}{2u^6}-\frac{7u_x^2\qs^1}{u^7}\right.\\
\notag
&\left.+\frac{u_{xxx}\qs}{3u^6}-\frac{4u_xu_{xx}\qs}{u^7}+\frac{7u_x^3\qs}{u^8}\right)+O(\qe^4),\\
\diff{\qs}{\tilde \qt_0} =&\,-\frac{\qs\qs^1}{u^3}+\qe^2\left(\frac{\qs\qs^3}{3u^6}-\frac{\qs^1\qs^2}{2u^6}-\frac{2u_x\qs\qs^2}{u^7}-\frac{2u_{xx}\qs\qs^1}{u^7}+\frac{7u_x^2\qs\qs^1}{u^8}\right)+O(\qe^4),\\
\diff{u}{\tilde \qt_1}=&\,\frac{\qs^1}{u}-\frac{u_x\qs}{2u^2}+\qe^2\left(-\frac{5\qs^3}{24u^4}+\frac{5u_x\qs^2}{4u^5}+\frac{3u_{xx}\qs^1}{4u^5}-\frac{23u_x^2\qs^1}{8u^6}\right.\\
\notag
&\left.+\frac{u_{xxx}\qs}{6u^5}-\frac{13u_xu_{xx}\qs}{8u^6}+\frac{19u_x^3\qs}{8u^7}\right)+O(\qe^4),\\
\label{AZ}
\diff{\qs}{\tilde \qt_1} =&\,-\frac{\qs\qs^1}{2u^2}+\qe^2\left(\frac{\qs\qs^3}{6u^5}-\frac{\qs^1\qs^2}{4u^5}-\frac{7u_x\qs\qs^2}{8u^6}-\frac{7u_{xx}\qs\qs^1}{u8^6}+\frac{21u_x^2\qs\qs^1}{8u^7}\right)+O(\qe^4).
\end{align}
Then we see that the flows $\diff{}{\tilde t}$, $\diff{}{\tilde \qt_0}$ and $\diff{}{\tilde \qt_1}$ defined by \eqref{AW}--\eqref{AZ} are now of $D$-type. More precisely, we have
\[
\diff{}{\tilde t} = D_{\tilde X},\quad \diff{}{\tilde \qt_0} = D_{\tilde P_0},\quad \diff{}{\tilde \qt_1} = D_{\tilde P_1},
\]
where the local functionals are given by
\begin{align*}
\tilde X&=\int \frac{u_x\qs}{u}+\qe^2\left(-\frac{u_{xxx}\qs}{12u^4}+\frac{3u_xu_{xx}\qs}{4u^5}-\frac{u_x^3\qs}{u^6}\right)+O(\qe^4),\\
\tilde P_0&=\frac12\int\frac{\qs\qs^1}{u^2}+\qe^2\left(-\frac{\qs\qs^3}{3u^5}+\frac{5u_x\qs\qs^2}{2u^6}+\frac{3u_{xx}\qs\qs^1}{2u^6}-\frac{7u_x^2\qs\qs^1}{u^7}\right)+O(\qe^4),\\
\tilde P_1&=\frac12\int\frac{\qs\qs^1}{u}+\qe^2\left(-\frac{5\qs\qs^3}{24u^4}+\frac{5u_x\qs\qs^2}{4u^5}+\frac{3u_{xx}\qs\qs^1}{4u^5}-\frac{23u_x^2\qs\qs^1}{8u^6}\right)+O(\qe^4).
\end{align*}
Hence the equation \eqref{AQ} is bihamiltonian, with the Hamiltonian operators given by
\begin{align*}
\tilde{\mathcal P}_0=&\,\frac{\qp_x}{u^2}-\frac{u_x}{u^3}+\qe^2\left(-\frac{\qp_x^3}{3u^5}+\frac{5u_x\qp_x^2}{2u^6}+\frac{3u_{xx}\qp_x}{2u^6}-\frac{7u_x^2\qp_x}{u^7}\right.\\
\notag
&\left.+\frac{u_{xxx}}{3u^6}-\frac{4u_xu_{xx}}{u^7}+\frac{7u_x^3}{u^8}\right)+O(\qe^4),\\
\tilde{\mathcal P}_1=&\,\frac{\qp_x}{u}-\frac{u_x}{2u^2}+\qe^2\left(-\frac{5\qp_x^3}{24u^4}+\frac{5u_x\qp_x^2}{4u^5}+\frac{3u_{xx}\qp_x}{4u^5}-\frac{23u_x^2\qp_x}{8u^6}\right.\\
\notag
&\left.+\frac{u_{xxx}}{6u^5}-\frac{13u_xu_{xx}}{8u^6}+\frac{19u_x^3}{8u^7}\right)+O(\qe^4).
\end{align*}
By using the definition \eqref{BA}, it is easy to see that the central invariant of the above bihamiltonian structure is again $\frac{1}{24}$.

We see that the key point of finding the bihamiltonian structure of the equation \eqref{AQ} is to find the new super variable $\qs$ such that the flows \eqref{AQ}--\eqref{AS} are of $D$-type. This is the key step of proving Theorem \ref{AE}.

\subsection{Proofs for the hydrodynamic cases}
We first prove Theorem \ref{AE} for the hydrodynamic bihamiltonian system. In this case the theorem has already been proved in \cite{xue2006bihamiltonian}. However, in order to prepare for the proof of the general cases, we need to re-prove this result using the $\qth$-formalism for Hamiltonian structures.

Let $(P_0,P_1)$ be a semisimple bihamiltonian structure, $(u^1,\dots,u^n)$ be its canonical coordinates and  $(\qth_1,\dots,\qth_n)$ be the corresponding odd variables. In terms of these coordinates, we have \cite{ferapontov2001compatible}
\begin{equation}
\label{norm-p}
P_0 = \frac 12\int \sum_{i,j}\left( \qd^{ij}f^i(u)\qth_i\qth_i^1+ A^{ij}\qth_i\qth_j\right),\ P_1 = \frac 12\int \sum_{i,j} \left(\qd^{ij}u^if^i(u)\qth_i\qth_i^1+ B^{ij}\qth_i\qth_j\right),
\end{equation}
where
\[
A^{ij} = \frac 12\left(\frac{f^i}{f^j}\diff{f^j}{u^i}u^{j,1}-\frac{f^j}{f^i}\diff{f^i}{u^j}u^{i,1}\right),\ B^{ij} = \frac 12\left(\frac{u^if^i}{f^j}\diff{f^j}{u^i}u^{j,1}-\frac{u^jf^j}{f^i}\diff{f^i}{u^j}u^{i,1}\right).
\]
Introduce the notation $f_i:=(f^i)^{-1}$, then we have rotation coefficients
\begin{equation}
\label{BC}
\qg_{ij}(u) = \frac{1}{2\sqrt{f_if_j}}\qp_jf_i,\quad i\neq j
\end{equation}
for the diagonal metric $\mathrm{diag}(f_1,\dots,f_n)$, here and henceforth we use $\qp_i$ to denote $\diff{}{u^i}$. The rotation coefficients satisfy the following equations \cite{ferapontov2001compatible}:
\begin{align}
\label{BB}
&\qp_k\qg_{ij} = \qg_{ik}\qg_{kj},\quad i,j,k\  \text{distinct},\\
&\qp_i\qg_{ij}+\qp_j\qg_{ji}+\sum_{k\neq i,j}\qg_{ki}\qg_{kj} = 0,\quad i\neq j,\\
\label{BD}
&u^i\qp_i\qg_{ij}+u^j\qp_j\qg_{ji}+\sum_{k\neq i,j}u^k\qg_{ki}\qg_{kj}+\frac12(\qg_{ij}+\qg_{ji}) = 0,\quad i\neq j.
\end{align} 
These equations are also defining equations for semisimple bihamiltonian structures. In another word, given $n$ nonzero functions $(f_1,\dots,f_n)$ and define $\qg_{ij}$ by \eqref{BC}, then the local functionals given by \eqref{norm-p} form a semisimple bihamiltonian structure if and only if the functions $\qg_{ij}$ satisfy the equations \eqref{BB}--\eqref{BD} (see the appendix of \cite{DLZ-1} for details).

Assume that we have a local functional $X\in\hm F^1$. It is proved in \cite{dubrovin2018bihamiltonian} that $X$ is a bihamiltonian vector field of $(P_0,P_1)$, i.e., $[X,P_0] = [X,P_1]=0$, if and only if 
\[
X = \int\sum_i A^i(u)u^{i,1}\qth_i,
\]
where the functions $A^i(u)$ satisfy the equations
\begin{equation}
\label{BE}
\qp_jA^i = \sqrt{\frac{f_j}{f_i}}\,\qg_{ij}(A^j-A^i),\quad i\neq j.
\end{equation}
Hence  a hydrodynamic bihamiltonian system must be of the form
\[
\diff{u^i}{t} = D_X(u^i)= A^iu^{i,1}.
\]
It admits the following super extension:
\begin{align}
\label{BF}
\diff{u^i}{t} &= D_X(u^i)= A^iu^{i,1},\\
\label{BG}
\diff{\qth_i}{t}&=D_X(\qth_i) = A^i\qth_i^1+\sum_{j\neq i}\left(\qp_jA^i\qth_i-\qp_iA^j\qth_j\right)u^{j,1},\\
\label{BH}
\diff{u^i}{\qt_0}&=D_{P_0}(u^i),\quad \diff{\qth_i}{\qt_0}=D_{P_0}(\qth_i),\\
\label{BI}
\diff{u^i}{\qt_1}&=D_{P_1}(u^i),\quad \diff{\qth_i}{\qt_1}=D_{P_1}(\qth_i).
\end{align}
We will need the following explicit expressions later:
\begin{align}
\label{DR}
\diff{u^i}{\qt_0}=&\, f^i\qth_i^1+\frac12\qp_if^i u^{i,1}\qth_i+\sum_{j\neq i}\left( a_{ji}u^{j,1}\qth_i+b_{ij}u^{j,1}\qth_j-b_{ji}u^{i,1}\qth_j\right),\\
\label{DS}
\diff{u^i}{\qt_1}=&\, u^if^i\qth_i^1+\frac12f^iu^{i,1}\qth_i+\frac12u^i\qp_if^i u^{i,1}\qth_i\notag\\
&+\sum_{j\neq i} \left(u^ia_{ji}u^{j,1}\qth_i+u^ib_{ij}u^{j,1}\qth_j-u^jb_{ji}u^{i,1}\qth_j\right),
\end{align}
here we denote for $i\neq j$
\begin{equation}
\label{BL}
a_{ij} = -f^j\sqrt{\frac{f^j}{f^i}}\,\qg_{ji},\quad b_{ij} = -\sqrt{f^if^j}\,\qg_{ji}.
\end{equation}

Assume all functions $A^i$ are non-zero, we can perform a linear reciprocal transformation by exchanging $x$ and $t$ in the equations \eqref{BF} and \eqref{BG}. By a straightforward computation, we obtain
\begin{align}
\label{BJ}
\diff{u^i}{\tilde t} &=\frac{1}{A^i}u^{i,1},\\
\label{BK}
\diff{\qth_i}{\tilde t}&=\frac{1}{A^i}\qth_i^1+\sum_{j\neq i}\frac{u^{j,1}}{A^iA^j}(\qp_iA^j\qth_j-\qp_jA^i\qth_i).
\end{align}
The  super extension of this system is obtained from \eqref{BH} and \eqref{BI} by using the substitution
\[
u^{i,1}\mapsto \frac{1}{A^i}u^{i,1},\quad \qth_i^1\mapsto \frac{1}{A^i}\qth_i^1+\sum_{j\neq i}\frac{u^{j,1}}{A^iA^j}(\qp_iA^j\qth_j-\qp_jA^i\qth_i).
\]
We use $\diff{}{\tilde\qt_0}$ and $\diff{}{\tilde\qt_1}$ to denote the odd flows of this super extended system, then they have the form
\begin{align}
\label{BM}
\diff{u^i}{\tilde \qt_0}=&\, f^i\left(\frac{1}{A^i}\qth_i^1+\sum_{j\neq i}\frac{u^{j,1}}{A^iA^j}(\qp_iA^j\qth_j-\qp_jA^i\qth_i)\right)\\
\notag
&+\frac12\qp_if^i\frac{u^{i,1}}{A^i}\qth_i+\sum_{j\neq i} \left(a_{ji}\frac{u^{j,1}}{A^j}\qth_i+b_{ij}\frac{u^{j,1}}{A^j}\qth_j-b_{ji}\frac{u^{i,1}}{A^i}\qth_j\right),\\
\label{BN}
\diff{u^i}{\tilde \qt_1}=&\, u^if^i\left(\frac{1}{A^i}\qth_i^1+\sum_{j\neq i}\frac{u^{j,1}}{A^iA^j}(\qp_iA^j\qth_j-\qp_jA^i\qth_i)\right)+\frac12 f^i\frac{u^{i,1}}{A^i}\qth_i\\
\notag
&+\frac12u^i\qp_if^i\frac{u^{i,1}}{A^i}\qth_i+\sum_{j\neq i}\left(u^ia_{ji}\frac{u^{j,1}}{A^j}\qth_i+u^ib_{ij}\frac{u^{j,1}}{A^j}\qth_j-u^jb_{ji}\frac{u^{i,1}}{A^i}\qth_j\right).
\end{align}

As we have explained in the beginning of this section, we need to define new odd variables $\qs_i$ such that the flow $\diff{}{t}$ defined by \eqref{BJ} and \eqref{BK} is of $D$-type. To this end, we introduce
\[
\qs_i = A^i\qth_i,
\]
then it follows from a direct computation that
\[
\diff{\qs_i}{\tilde t} = \frac{1}{A^i}\qs_i^1+\sum_{j\neq i} u^{j,1}\left(\qp_j\frac{1}{A^i}\qs_i-\qp_i\frac{1}{A^j}\qs_j\right).
\]
Then we see that in terms of the coordinates $(u^i,\qs_i)$, the flow $\diff{}{\tilde t}$ can be written as
\[
\diff{}{\tilde t} = D_{\tilde X},\quad \tilde X = \int\sum_i\frac{1}{A^i}u^{i,1}\qs_i,
\]
thus it is of $D$-type. It remains to show that the flows $\diff{}{\tilde\qt_0}$ and $\diff{}{\tilde\qt_1}$ are also of $D$-type in terms of coordinates $(u^i,\qs_i)$. This can be done by using the following lemma.

\begin{Lem}\label{BO}
Given a derivation $\diff{}{\qt}\in\derx^1_1$ which satisfies the condition $[\diff{}{\qt},\diff{}{\qt}] = 0$, then we have the flow of the form
\[
\diff{u^i}{\qt} = \sum_jg^{ij}\qth_j^1+\sum_{j,k}\Qg^{ij}_ku^{k,1}\qth_j
\]
in terms of some local coordinates $(u^i,\qth_i)$. If $\det(g^{ij})\neq 0$ and $(g_{ij}) := (g^{ij})^{-1}$ defines a flat metric with the Christoffel symbols of its Levi-Civita connection given by $\Qg^i_{jk} =- \sum_{\ell}g_{j\ell}\Qg^{\ell i}_k$, then $\diff{}{\qt}$ is of $D$-type, i.e. there is a local fcuntional 
\[
P = \frac12\int\sum_{i,j}g^{ij}\qth_i\qth_j^1+\sum_{i,j,k}\Qg^{ij}_ku^{k,1}\qth_i\qth_j
\]
 such that
\[
\diff{}{\qt} = D_P.
\]
\end{Lem}
\begin{proof}
Consider the following Miura-type transformation from coordinates $(u^i,\qth_i)$ to $(\tilde u^i,\tilde \qth_i)$:
\[
\tilde u^i = \tilde u^i(u),\quad \tilde \qth_i = \sum_j\diff{u^j}{\tilde u^i}\qth_j.
\]
By a straightforward computation, we have
\[
\diff{\tilde u^i}{\qt} = \sum_j\tilde g^{ij}\tilde\qth_j^1+\sum_{j,k}\tilde \Qg^{ij}_{k}\tilde u^{k,1}\tilde\qth_j,
\]
where the functions $\tilde g^{ij}$ and $\tilde \Qg^{ij}_{k}$ are given by
\begin{align*}
\tilde g^{ij} &= \sum_{k,\ell}\diff{\tilde u^i}{u^k}\diff{\tilde u^j}{u^\ell}g^{k\ell},\\
\tilde \Qg^{ij}_{k}&=-\sum_\ell \tilde g^{i\ell}\left(\sum_p\diff{\tilde u^j}{ u^p}\frac{\qp^2 u^p}{\qp\tilde u^\ell\tilde u^k}+\sum_{p,q,r}\diff{u^p}{\tilde u^\ell}\diff{u^q}{\tilde u^k}\diff{\tilde u^j}{u^r}\Qg^r_{pq}\right).
\end{align*}
Therefore  if $(g_{ij})$ defines a flat metric and  $\Qg^i_{jk}$ are the Christoffel symbols of its Levi-Civita connection, then by choosing $\tilde u^i$ to be its flat coordinates we may assume that $\tilde g^{ij}$ are constants and all the corresponding Christoffel symbols vanish. Thus we have
\[
\diff{\tilde u^i}{\qt} = \sum_j\tilde g^{ij}\tilde\qth_j^1.
\]
By using the fact that 
\[
\left[\diff{}{\qt},\diff{}{\qt}\right]\tilde u^i = 2\sum_j\tilde g^{ij}\qp_x\diff{\tilde \qth_j}{\qt} = 0,
\]
it follows that
\[
\diff{\tilde \qth_i}{\qt} = 0.
\]
Therefore in terms of coordinates $(\tilde u^i,\tilde \qth_i)$, the flow $\diff{}{\qt}$ is of $D$-type, i.e., it can be represented as
\[
\diff{}{\qt} = D_{\tilde P},\quad \tilde P = \frac12\int\sum_{ij}\tilde g^{ij}\tilde\qth_i\tilde\qth_j^1.
\]
Since the property of being $D$-type is preserved under Miura-type transformations, we conclude that in terms of the original coordinates $(u^i,\qth_i)$, $\diff{}{\qt}$ is also of $D$-type. The lemma is proved.
\end{proof}

Let us come back to prove that the flows $\diff{}{\tilde\qt_0}$ and $\diff{}{\tilde\qt_1}$ are of $D$-type in terms of coordinates $(u^i,\qs_i)$. We first write down the flows $\diff{u^i}{\tilde\qt_0}$ and $\diff{u^i}{\tilde\qt_1}$ by  substituting
\[
\qth_i\mapsto\frac{\qs_i}{A^i}
\]
into equations \eqref{BM} and \eqref{BN} to obtain
\begin{align*}
\diff{u^i}{\tilde \qt_0}=&\, \frac{f^i}{A^i}\left(\frac{\qs_i^1}{A^i}-\frac{\qp_iA^i}{(A^i)^2}u^{i,1}\qs_i-\sum_{j\neq i}\frac{1}{(A^i)^2}\sqrt{\frac{f_j}{f_i}}\,\qg_{ij}(A^j-A^i)u^{j,1}\qs_i\right)\\
&+f^i\sum_{j\neq i}\frac{u^{j,1}}{A^iA^j}\left(\sqrt{\frac{f_i}{f_j}}\,\qg_{ji}(A^i-A^j)\frac{\qs_j}{A^j}-\sqrt{\frac{f_j}{f_i}}\,\qg_{ji}(A^j-A^i)\frac{\qs_i}{A^i}\right)\\
\notag
&+\frac12\qp_if^i \frac{u^{i,1}}{A^i}\frac{\qs_i}{A^i}+\sum_{j\neq i} \left(a_{ji}\frac{u^{j,1}}{A^j}\frac{\qs_i}{A^i}+b_{ij}\frac{u^{j,1}}{A^j}\frac{\qs_j}{A^j}-b_{ji}\frac{u^{i,1}}{A^i}\frac{\qs_j}{A^j}\right),\\
\diff{u^i}{\tilde \qt_1}=&\, \frac{u^if^i}{A^i}\left(\frac{\qs_i^1}{A^i}-\frac{\qp_iA^i}{(A^i)^2}u^{i,1}\qs_i-\sum_{j\neq i}\frac{1}{(A^i)^2}\sqrt{\frac{f_j}{f_i}}\,\qg_{ij}(A^j-A^i)u^{j,1}\qs_i\right)+\frac12 f^i\frac{u^{i,1}}{A^i}\frac{\qs_i}{A^i}\\
&+u^if^i\sum_{j\neq i}\frac{u^{j,1}}{A^iA^j}\left(\sqrt{\frac{f_i}{f_j}}\,\qg_{ji}(A^i-A^j)\frac{\qs_j}{A^j}-\sqrt{\frac{f_j}{f_i}}\,\qg_{ji}(A^j-A^i)\frac{\qs_i}{A^i}\right)\\
\notag
&+\frac12u^i\qp_if^i \frac{u^{i,1}}{A^i}\frac{\qs_i}{A^i}+\sum_{j\neq i} \left(u^ia_{ji}\frac{u^{j,1}}{A^j}\frac{\qs_i}{A^i}+u^ib_{ij}\frac{u^{j,1}}{A^j}\frac{\qs_j}{A^j}-u^jb_{ji}\frac{u^{i,1}}{A^i}\frac{\qs_j}{A^j}\right),
\end{align*}
here we used the equations \eqref{BE}. According to Lemma \ref{BO}, in order to prove that the above flows are of $D$-type, we only need to show that the (contravariant) metrics 
\[\mathrm{diag}(\tilde f^1,\dots,\tilde f^n),\quad \mathrm{diag}(u^1\tilde f^1,\dots,u^n\tilde f^n),\quad \textrm{with}\ \tilde f^i = \frac{f^i}{(A^i)^2} \]
are flat, and the above flows can be represented in the form
\begin{align}
\label{BP}
\diff{u^i}{\tilde\qt_0}=&\, \tilde f^i\qs_i^1+\frac12 \qp_i\tilde f^iu^{i,1}\qs_i+\sum_{j\neq i}\left(\tilde a_{ji}u^{j,1}\qs_i+\tilde b_{ij}u^{j,1}\qs_j-\tilde b_{ji}u^{i,1}\qs_j\right),\\
\label{BQ}
\diff{u^i}{\tilde\qt_1}=&\, u^i\tilde f^i\qs_i^1+\frac12\tilde f^iu^{i,1}\qs_i+\frac12 u^i\qp_i\tilde f^iu^{i,1}\qs_i\\
&+\sum_{j\neq i}\left( u^i\tilde a_{ji}u^{j,1}\qs_i
+u^i\tilde b_{ij}u^{j,1}\qs_j-u^j\tilde b_{ji}u^{i,1}\qs_j\right).
\end{align}
Here
\begin{equation*}
\tilde a_{ij} = -\tilde f^j\sqrt{\frac{\tilde f^j}{\tilde f^i}}\,\tilde \qg_{ji},\quad \tilde b_{ij} = -\sqrt{\tilde f^i\tilde f^j}\,\tilde \qg_{ji},
\end{equation*}
and $\tilde \qg_{ij}$ are rotation coefficients of the metric $(\tilde f_1,\dots,\tilde f_n)$ with $\tilde f_i = (\tilde f^i)^{-1}$.

By a straightforward computation using the definition \eqref{BC} and the equations \eqref{BE}, one can prove that
\[
\tilde \qg_{ij} = \frac{1}{2\sqrt{\tilde f_i\tilde f_j}}\qp_j\tilde f_i = \qg_{ij}.
\]
In particular it follows that functions $\tilde \qg_{ij}$ satisfy the equations \eqref{BB}--\eqref{BD}, hence the metrics
$\mathrm{diag}(\tilde f^1,\dots,\tilde f^n)$ and $\mathrm{diag}(u^1\tilde f^1,\dots,u^n\tilde f^n)$ are flat. Again by computing directly, it is easy to see that the equations \eqref{BP} and \eqref{BQ} hold true. 

We summarize the results obtained so far in the following theorem.
\begin{Th}
\label{BR}
Let  $(P_0,P_1)$ be a semisimple bihamiltonian structure of hydrodynamic type of the form \eqref{norm-p}, and $\diff{}{t}$ be any of its bihamiltonian vector field which has the form
\[
\diff{u^i}{t} = A^iu^{i,1},\quad i = 1,\dots,n,
\]
here $u^i$ are  canonical coordinates of $(P_0,P_1)$.
Assume that $A^i\neq 0$, then the system of evolutionary  PDEs
\[
\diff{u^i}{t} = \frac{1}{A^i}u^{i,1}\quad i = 1,\dots,n
\]
possesses a bihamiltonian structure given by the local funtionals
\begin{equation}
\label{DV}
\tilde P_0 = \frac 12\int \sum_{i,j} \left(\qd^{ij}\tilde f^i(u)\qth_i\qth_i^1+ \tilde A^{ij}\qth_i\qth_j\right),\ \tilde P_1 = \frac 12\int \sum_{i,j} \left(\qd^{ij}u^i\tilde f^i(u)\qth_i\qth_i^1+\tilde  B^{ij}\qth_i\qth_j\right),
\end{equation}
where\[
\tilde f^i = \frac{f^i}{(A^i)^2}
\]
and
\[
\tilde A^{ij} = \frac 12\left(\frac{\tilde f^i}{\tilde f^j}\diff{\tilde f^j}{u^i}u^{j,1}-\frac{\tilde f^j}{\tilde f^i}\diff{\tilde f^i}{u^j}u^{i,1}\right),\ \tilde B^{ij} = \frac 12\left(\frac{u^i\tilde f^i}{\tilde f^j}\diff{\tilde f^j}{u^i}u^{j,1}-\frac{u^j\tilde f^j}{\tilde f^i}\diff{\tilde f^i}{u^j}u^{i,1}\right).
\]
Moreover, the metrics $\mathrm{diag}(\tilde f^1,\dots,\tilde f^n)$ and $\mathrm{diag}(f^1,\dots,f^n)$ share the same rotation coefficients.
\end{Th}

To complete the proof of Theorem \ref{AE} for the hydrodynamic case, we need to pick a bihamiltonian vector field $\diff{}{s}$ which is a symmetry of the system \eqref{BF} and can be represented as 
\begin{equation}
\label{DW}
\diff{u^i}{s} = B^iu^{i,1}.
\end{equation}
After performing the linear reciprocal transformation corresponding to $\diff{}{t}$, the flow $\diff{}{s}$ is transformed to
\begin{equation}
\label{DX}
\diff{u^i}{\tilde s} = \frac{B^i}{A^i}u^{i,1}.
\end{equation}
We need to show that this flow shares the same bihamiltonian structure with the flow \eqref{BJ}. By applying Theorem \ref{BR} and using the equation \eqref{BE}, we only need to check the validity of  the identities
\[
\qp_j\frac{B^i}{A^i} = \sqrt{\frac{\tilde f_j}{\tilde f_i}}\,\qg_{ij}(\frac{B^j}{A^j}-\frac{B^i}{A^i}),\quad i\neq j,
\]
which are easy to verify. Moreover, it is also easy to verify by a direct calculation that in terms of coordinates $(u^i,\qs_i)$, the flow $\diff{}{\tilde s}$ is of $D$-type. Hence Theorem \ref{AE} holds true for the hydrodynamic case.

We can further describe transformation rules of Hamiltonians under linear reciprocal transformations. Let us write the flow \eqref{DW} into the Hamiltonian form
\[
\diff{u^i}{s} = B^iu^{i,1} = \sum_j\mathcal P_0^{ij}\vard{H}{u^j},
\]
where $(\mathcal P_0^{ij})$ is the Hamiltonian operator of the Hamiltonian structure $P_0$ given in \eqref{norm-p} and $H = \int h$, where $h\in\hm A^0_0$. By performing  the linear reciprocal transformation with respect to the flow \eqref{BF} to \eqref{DW}, we arrive at the bihamiltonian vector field \eqref{DX}. We know that there exists $G\in\hm F^0_0$ such that
\[
\diff{u^i}{\tilde s} = \frac{B^i}{A^i}u^{i,1} = \sum_j\tilde{\mathcal P}_0^{ij}\vard{G}{u^j},
\]
where $(\tilde{\mathcal P}_0^{ij})$ is the Hamiltonian operator of the Hamiltonian structure $\tilde P_0$ given in \eqref{DV}. Let $G = \int g$ for $g\in\hm A^0_0$, then we have the following result.
\begin{Th}
\label{ham-reci}
The function $g$ can be chosen to be a solution of the following system of PDEs:
\[
\diff{g}{u^i} = A^i\diff{h}{u^i},\quad i = 1,\dots,n.
\]
\end{Th}
\begin{proof}
The identity
\[ B^iu^{i,1} = \sum_j\mathcal P_0^{ij}\vard{H}{u^j}\] yields the following identities satisfied by the function $h$:
\[
\qp_i\qp_jh = \sqrt{\frac{f^i}{f^j}}\,\qg_{ij}\qp_ih+\sqrt{\frac{f^j}{f^i}}\,\qg_{ji}\qp_jh\]
for $i\neq j$ and we also have
\[
B^i = f^i\qp_i^2h+\frac12\qp_if^i\qp_ih+\sum_{j\neq i}\sqrt{f^if^j}\qg_{ij}\qp_jh.
\]
Then the theorem is proved by a direct computation by using the above identities.
\end{proof}

\subsection{Proofs for the general case}
We proceed to prove Theorem \ref{AE} for the general case. As we have already seen from the last subsection, the fundamental idea to prove this theorem is to find suitable transformations of odd variables. Let us first  prove the following theorem.
\begin{Th}
\label{BS}
Let the vector fields $\diff{}{t}\in\derx^0_{\geq 1}$ and $\diff{}{\qt_0},\diff{}{\qt_1}\in\derx^1_{\geq 1}$ satisfy the conditions
\[
\left[\diff{}{t},\diff{}{\qt_a}\right] =\left[\diff{}{\qt_a},\diff{}{\qt_b}\right]= 0,\quad a,b = 0,1,
\]
and admit the following decomposition:
\begin{align*}
\diff{}{t} &= \diff{}{t^{[0]}}+\sum_{k\geq 2}\diff{}{t^{[k]}},\quad \diff{}{t^{[k]}}\in\derx^0_{k+1},\quad k\geq 0,\\
\diff{}{\qt_a} &= \diff{}{\qt_a^{[0]}}+\sum_{k\geq 2}\diff{}{\qt_a^{[k]}},\quad \diff{}{\qt_a^{[k]}}\in\derx^1_{k+1},\quad k\geq 0,\quad a = 0,1.
\end{align*}
Here, in terms of local coordinates $(u^i,\qth_i)$, the leading terms
$\diff{}{t^{[0]}}$, $\diff{}{\qt_0^{[0]}}$ and  $\diff{}{\qt_1^{[0]}}$ are of $D$-type. Denote
\[
\diff{}{\qt_a^{[0]}} = D_{P_a^{[0]}},\quad P_a^{[0]}\in\hm F^2_1,\quad a = 0,1,
\]
then if $(P_0^{[0]},P_1^{[0]})$ is a semisimple bihamiltonian structure of hydrodynamic type, there exist transformations of the odd variables
\[
\qs_i = \qth_i+W_i,\quad W_i\in\hm A^1_{\geq 2}
\]
such that, in terms of local coordinates $(u^i,\qs_i)$, the vector fields $\diff{}{t}$, $\diff{}{\qt_0}$ and  $\diff{}{\qt_1}$ are also of $D$-type.
\end{Th}

It is easy to see that the first part of Theorem \ref{AE} follows from Theorem \ref{BS}. Indeed, consider a bihamiltonian structure $(P_0,P_1)$ and any of its bihamiltonian vector field $\diff{}{t}$. We assume that the leading term of $(P_0,P_1)$ is a semisimple bihamiltonian structure of hydrodynamic type. By applying a suitable Miura-type transformation if necessary, we may require that the $\deg_x = 2$ components of $(P_0,P_1)$ and $\diff{}{t}$ vanish \cite{DLZ-1}. After performing the linear reciprocal transformation with respect to $\diff{}{t}$, we obtain the vector fields $\diff{}{\tilde t}$, $\diff{}{\tilde \qt_0}$ and $\diff{}{\tilde \qt_1}$. Firstly, by using the results of Theorem \ref{BR}, we may assume that the leading terms of $\diff{}{\tilde t}$, $\diff{}{\tilde \qt_0}$ and $\diff{}{\tilde \qt_1}$ are of $D$-type. Since the $\deg_x = 2$ components of these vector fields  vanish, we may apply Theorem \ref{BS} to find new odd variables such that, in terms of the new odd variables, these vector fields are also of $D$-type. Therefore we may represent the odd flows as
\[
\diff{}{\tilde \qt_a} = D_{\tilde P_a},\quad \tilde P_a\in\hm F^2_{\geq 1},\quad a = 0,1.
\]
In particular, $\diff{}{\tilde t}$ is a bihamiltonian vector field of $(\tilde P_0,\tilde P_1)$.

To prove the second part of Theorem \eqref{AE}, we take another bihamiltonian vector field $\diff{}{s}$ of $(P_0,P_1)$, which is transformed to $\diff{}{\tilde s}$ by the linear reciprocal transformation with respect to $\diff{}{t}$. We need to prove that $\diff{}{\tilde s}$ is a bihamiltonian vector field with respect to $(\tilde P_0,\tilde P_1)$ in terms of new odd variables described in Theorem \ref{BS}. By applying Theorem \eqref{AE} for the hydrodynamic case, we know that the leading term of  $\diff{}{\tilde s}$ is a bihamiltonian vector field with respect to the leading term of $(\tilde P_0,\tilde P_1)$. According to the result of \cite{DLZ-1}, this leading term determines a unique bihamiltonian vector field of $(\tilde P_0,\tilde P_1)$, i.e., there exists $X\in \hm F^1$ such that $[X,\tilde P_0]=[ X,\tilde P_1]=0$ and the leading term of $D_{X}$ coincides with that of $\diff{}{\tilde s}$. Therefore it follows that the vector field \[\diff{}{\tilde s}-D_X\in\derx^0_{\geq 2}\] commutes with $D_{\tilde P_0}$ and $D_{\tilde P_1}$. By applying the result  \[VBH^0_{\geq 2}(\derx) =\derx^0_{\geq 2}\cap\ker d_{\tilde P_0}\cap\ker d_{\tilde P_1}= 0\]
of Theorem \ref{AK}, we conclude that $\diff{}{\tilde s} = D_X$. So the second part of Theorem \eqref{AE} is proved. By a further computation, we can check that the central invariants of $(\tilde P_0,\tilde P_1)$ coincide with those of $(P_0,P_1)$. Hence Theorem \ref{AE} is proved.

In the remaining part of this subsection, we prove Theorem \ref{BS}. Instead of finding suitable transformations of odd variables, it is much more convenient to consider the following general transformations on $\hm A$:
\begin{align}
\label{BT}
\tilde u^{i,s}&= u^{i,s}+\qp_x^sZ^i,\quad Z^i\in\hm A^0_{\geq 1},\\
\label{BU}
\tilde\qth_i^s&= \qth_i^s+\qp_x^sW_i,\quad W_i\in\hm A^1_{\geq 1}.
\end{align}
Note that such a transformation can always be decomposed into a Miura-type transformation coupled with a transformation of odd variables, and Miura-type transformations preserve the property of being $D$-type of a vector field, therefore we may freely consider such general transformations when proving Theorem \ref{BS}.

As a first step towards proving Theorem \ref{BS}, we give an alternative description of transformations \eqref{BT} and \eqref{BU}, and describe their actions on the space $\derx$.
\begin{Lem}
\label{DQ}
Given a transformation of the form \eqref{BT} and \eqref{BU}, there exists a vector field $T\in\derx^0_{\geq 1}$, such that
\[
\tilde u^i = \exp(T)u^i,\quad \tilde \qth_i = \exp(T)\qth_i.
\]
Moreover, a vector field $X\in\derx$ is transformed to $\tilde X$ after performing the transformations \eqref{BT} and \eqref{BU}, where
\[
\tilde X = \exp(-\mathrm{ad}_T)X.
\]
\end{Lem}
\begin{proof}
This lemma is easy to prove by slightly modifying the proof of Lemma 2.5.6 and Theorem 2.5.7 given in \cite{liu2011jacobi} and we omit the details here. The lemma is proved.
\end{proof}

\begin{Lem}
\label{BV}
Consider the vector field $\diff{}{\qt}\in\derx^1_{\geq 1}$ which satisfies the condition
\[
\left[\diff{}{\qt},\diff{}{\qt}\right] =0,
\]
and has the decomposition
\begin{align*}
\diff{}{\qt} &= \sum_{k\geq 0}\diff{}{\qt^{[k]}},\quad \diff{}{\qt^{[k]}}\in\derx^1_{k+1}.
\end{align*}
Assume that in terms of local coordinates $(u^i,\qth_i)$ the vector field
$\diff{}{\qt^{[0]}}$ is of $D$-type, i.e., it can be represented as
\[
\diff{}{\qt^{[0]}} = D_{P^{[0]}},\quad P^{[0]}\in\hm F^2_1.
\] 
If $P^{[0]}$ is a Hamiltonian structure of hydrodynamic type, then there exists a transformation 
\[
\qs_i = \qth_i+W_i,\quad W_i\in\hm A^1_{\geq 2}
\]
of the odd variables such that, in terms of local coordinates $(u^i,\qs_i)$, the vector field $\diff{}{\qt}$ is of $D$-type.
\end{Lem}
\begin{proof}
Let us consider a transformation
\[
\tilde u^i = \exp(T)u^i,\quad \tilde \qth_i = \exp(T)\qth_i
\]
with a certain $T\in\derx^0_{\geq 1}$, then Lemma \ref{DQ} implies that the vector field $\diff{}{\qt}$ is transformed to
\[
\diff{}{\tilde \qt} = \exp(-ad_T)\diff{}{\qt} .
\]
If we can find such a $T$ satisfying the condition
\[
\diff{}{\tilde \qt} = \diff{}{\qt^{[0]}},
\]
then it follows that $\diff{}{\tilde \qt}$ is of $D$-type in terms of the new variables $(\tilde u^i,\tilde\qth_i)$.

Let us find $T$ recursively by induction. Firstly, the condition $[\diff{}{\qt},\diff{}{\qt}]=0$ implies that
\[
\left[\diff{}{\qt^{[1]}},D_{P^{[0]}}\right] = 0.
\]
Therefore $\diff{}{\qt^{[1]}}$ defines a cohomology class in $H^1_{2}(\derx,P^{[0]})$. By applying Theorem \ref{AK} we know that such a cohomology class is trivial, so there exists $T^{[1]}\in\derx^0_1$ such that
\[
\diff{}{\qt^{[1]}} = \fk{T^{[1]}}{D_{P^{[0]}}} = \fk{T^{[1]}}{\diff{}{\qt^{[0]}}}.
\]
Then by performing the transformation
\[
u^i\mapsto \exp(T^{[1]})u^i,\quad \qth_i\mapsto\exp(T^{[1]})\qth_i,
\]
the vector field $\diff{}{\qt}$ is transformed to
\[
\diff{}{\qt}\mapsto\diff{}{\tilde \qt} = \diff{}{\qt^{[0]}}+\sum_{k\geq 2}\diff{}{\tilde \qt^{[k]}},\quad \diff{}{\tilde \qt^{[k]}}\in\derx^1_{k+1}.
\]
Since the vector field $\diff{}{\tilde\qt}$ still satisfies the condition
\[
\fk{\diff{}{\tilde\qt}}{\diff{}{\tilde\qt}} = 0,
\]
by a similar argument we can find a transformation to eliminate the term $\diff{}{\tilde \qt^{[2]}}$ of $\diff{}{\tilde \qt}$. Continuing this procedure in a  recursive way, we arrive at a desired transformation such that
\[
\diff{}{\qt}\mapsto\diff{}{\qt^{[0]}}.
\]
The lemma is proved.
\end{proof}

Now we are ready to prove Theorem \ref{BS}.
\begin{proof}[Proof of Theorem \ref{BS}.] The idea is the same as that of the proof of Lemma \ref{BV}. We are to find a transformation of the form \eqref{BT} and \eqref{BU} such that the flows $\diff{}{\qt_0}$ and $\diff{}{\qt_1}$ are transformed to vector fields of $D$-type. By applying Lemma \ref{BV}, we may assume that $\diff{}{\qt_0} = \diff{}{\qt_0^{[0]}}$. Then the condition
\begin{equation}
\label{BW}
\fk{\diff{}{\qt_a}}{\diff{}{\qt_b}} = 0,\quad a,b = 0
\end{equation}
implies that
\[
\fk{\diff{}{\qt_1^{[2]}}}{D_{P_0^{[0]}}} =\fk{\diff{}{\qt_1^{[2]}}}{D_{P_1^{[0]}}} = 0,
\]
therefore $\diff{}{\qt_1^{[2]}}$ defines a cohomology class of $VBH^1_{3}(\derx, P_0^{[0]},P_1^{[0]})$. It follows from \eqref{AM} the existence of $S\in\hm F^2_3$ and $Y\in\derx^{-1}_1$ such that
\[
\diff{}{\qt_1^{[2]}} = D_S+\fk{\diff{}{\qt_1^{[0]}}}{\fk{\diff{}{\qt_0^{[0]}}}{Y}}.
\]
Therefore if we define
\[
T^{[2]} = \fk{\diff{}{\qt_0^{[0]}}}{Y}\in\derx^0_2,
\]
then after the transformation
\[
u^i\mapsto \exp(T^{[2]})u^i,\quad \qth_i\mapsto \exp(T^{[2]})\qth_i,
\]
the vector fields $\diff{}{\qt_0}$ and $\diff{}{\qt_1}$ are transformed to
\[
\diff{}{\qt_0}\mapsto \diff{}{\qt_0},\quad \diff{}{\qt_1}\mapsto D_{P_1^{[0]}}+D_S+\sum_{k\geq 3}\diff{}{\tilde \qt_1^{[k]}},\quad \diff{}{\tilde \qt_1^{[k]}}\in\derx^1_{k+1}.
\]

We proceed to find a transformation of the form \eqref{BT} and \eqref{BU} which turns the vector field $\diff{}{\tilde \qt_1^{[3]}}$ into $D$-type. Recall that the local functional $S$ has the expression 
\[
S = \fk{P_0^{[0]}}{\left[P_1^{[0]},\int\sum_i c_i(u^i)u^{i,1}\log u^{i,1}\right]},
\]
so we have \[\left[S,P_0^{[0]}\right] = \left[S,P_1^{[0]}\right] = 0.\] Then by using the theory of bihamiltonian cohomology \cite{carlet2018central,carlet2018deformations,liu2013bihamiltonian}, we can find $R\in\hm F^2_4$ such that
\[
\left[R,P_0^{[0]}\right] = 0,\quad \left[R,P_1^{[0]}\right] = -\frac12 [S,S].
\]
Since the condition \eqref{BW} yields the relations
\[
\fk{\diff{}{\tilde \qt_1^{[3]}}}{D_{P_0}^{[0]}} = 0,\quad \fk{\diff{}{\tilde \qt_1^{[3]}}}{D_{P_1}^{[0]}} = -\frac 12[D_S,D_S], 
\]
by applying the identity \eqref{AJ} we see that the vector field 
\[
\diff{}{\tilde \qt_1^{[3]}}-D_R
\]
defines a cohomology class of $VBH^1_{4}(\derx, P_0^{[0]},P_1^{[0]})$. Hence it follows from Theorem \ref{AK} the existence of $\tilde Y\in\derx^{-1}_2$ such that
\[
\diff{}{\tilde \qt_1^{[3]}}-D_R = \fk{\diff{}{\qt_1^{[0]}}}{\fk{\diff{}{\qt_0^{[0]}}}{\tilde Y}}.
\]
Let us define
\[
T^{[3]} = \fk{\diff{}{\qt_0^{[0]}}}{\tilde Y}\in\derx^0_3,
\]
then after the transformation 
\[
u^i\mapsto \exp(T^{[3]})u^i,\quad \qth_i\mapsto \exp(T^{[3]})\qth_i,
\]
the vector fields $\diff{}{\tilde \qt_0}$ and $\diff{}{\tilde \qt_1}$ are transformed as
\[
\diff{}{\tilde \qt_0}\mapsto \diff{}{\tilde \qt_0},\quad \diff{}{\tilde \qt_1}\mapsto D_{P_1^{[0]}}+D_S+D_R+\dots.
\]
By continuing this procedure recursively, we arrive at a transformation which turns $\diff{}{\qt_{0}}$ and $\diff{}{\qt_1}$ into $D$-type.

Finally we prove that after $\diff{}{\qt_{0}}$ and $\diff{}{\qt_1}$ are transformed to vector fields of $D$-type, the vector field $\diff{}{t}$ is automatically turned into $D$-type. To avoid complicated notations, we assume that $\diff{}{\qt_{0}}$ and $\diff{}{\qt_1}$ are already of $D$-type. By the assumption of the theorem, the leading term of $\diff{}{t}$ is a vector field of $D$-type, i.e., we have
\[
\diff{}{t^{[0]}} = D_{X^{[0]}},\quad X^{[0]}\in\hm F^1_1.
\]
It follows from Theorem \ref{BX} the existence of  $X\in\hm F^1_{\geq 1}$ with leading term $X^{[0]}$ such that
\[
\fk{D_X}{\diff{}{\qt_0}} = \fk{D_X}{\diff{}{\qt_1}} = 0.
\]
Then we conclude that 
\[
\fk{\diff{}{t}-D_X}{\diff{}{\qt_0}} = \fk{\diff{}{t}-D_X}{\diff{}{\qt_1}} = 0,
\]
hence it follows from $VBH^0_{\geq 2}(\derx) = 0$ that $\diff{}{t} = D_X$. The theorem is proved.
\end{proof}

Finally, let us consider the central invariants of the bihamiltonian structure after the linear reciprocal transformation. Assume we have a bihamiltonian structure $(P_0,P_1)$ whose leading term forms a semisimple bihamiltonian structure of hydrodynamic type.   After performing a Miura-type transformation if necessary, we may assume $P_0 = P_0^{[0]}\in\hm F^2_1$ and
\[
P_1 = P_1^{[0]}+R,\quad P_1^{[0]}\in\hm F^2_1,\quad R\in\hm F^2_{\geq 3},
\]
where $(P_0^{[0]},P_1^{[0]})$ is a semisimple bihamiltonian structure of the form \eqref{norm-p}. Let $X\in\hm F^1_{\geq 1}$ be a bihamiltonian vector field with respect to $(P_0,P_1)$. Let us denote by $\diff{}{t}$ the vector field $D_X$, then in terms of the canonical coordinates of $(P_0^{[0]},P_1^{[0]})$ we have
\begin{align}
\label{BY}
\diff{u^i}{t}:=&\,D_X(u^i) = A^iu^{i,1}+\sum_{j}B^i_ju^{j,3}+\dots,\\
\label{BZ}
\diff{\qth_i}{t}:=&\,D_X(\qth_i)=A^i\qth_i^1+\sum_{j\neq i}\left(\qp_jA^i\qth_i-\qp_iA^j\qth_j\right)u^{j,1}+\sum_j B^i_j\qth_j^3+\dots,
\end{align}
here and henceforth we omit all the terms that are irrelevant to the computations below. Assume that the central invariants of $(P_0,P_1)$ are $c_1(u^1),\dots,c_n(u^n)$, then we the corresponding odd flows have the form
\begin{align*}
\diff{u^i}{\qt_0}=&\, f^i\qth_i^1+\frac12\qp_if^iu^{i,1}\qth_i+\sum_{j\neq i}\left( a_{ji}u^{j,1}\qth_i+b_{ij}u^{j,1}\qth_j-b_{ji}u^{i,1}\qth_j\right),\\
\diff{u^i}{\qt_1}=&\, u^if^i\qth_i^1+\frac12f^iu^{i,1}\qth_i+\frac12u^i\qp_if^iu^{i,1}\qth_i+\sum_{j\neq i}\left( u^ia_{ji}u^{j,1}\qth_i+u^ib_{ij}u^{j,1}\qth_j-u^jb_{ji}u^{i,1}\qth_j\right)\\&+3c_i(u^i)(f^i)^2\qth_i^3+\dots,
\end{align*}
where the leading terms of the above flows are given by \eqref{DR} and \eqref{DS}, and the $\qth_i^3$ coefficient of the $\deg_x  = 3$ component of $\diff{u^i}{\qt_1}$ follows from the definition \eqref{BA} of the central invariants. 
We perform the linear reciprocal transformation with respect to the flows \eqref{BY} and \eqref{BZ} to arrive at the following new flows:
\begin{align}
\diff{u^i}{\tilde t}&=\frac{1}{A^i}u^{i,1}-\sum_j \frac{B^i_j}{A^i(A^j)^3}u^{j,3}+\dots,\\
\diff{\qth_i}{\tilde t}&=\frac{1}{A^i}\qth_i^1+\sum_{j\neq i}\frac{u^{j,1}}{A^iA^j}(\qp_iA^j\qth_j-\qp_jA^i\qth_i)-\sum_j\frac{B^i_j}{A^i(A^j)^3}u^{j,3}+\dots.
\end{align}
Then the flows $\diff{}{\qt_0}$ and $\diff{}{\qt_1}$ are transformed to
\begin{align*}
\diff{u^i}{\tilde \qt_0} &= \frac{f^i}{A^i}\qth_{i}^1-f^i\sum_j\frac{B^i_j}{A^i(A^j)^3}u^{j,3}+\dots,\\
\diff{u^i}{\tilde \qt_1} &= u^i\frac{f^i}{A^i}\qth_{i}^1-u^if^i\sum_j\frac{B^i_j}{A^i(A^j)^3}u^{j,3}+3c_i(u^i)(f^i)^2\frac{1}{(A^i)^3}\qth_i^3+\dots.
\end{align*}
According to Theorem \ref{BS}, there exist new odd variables
\[
\qs_i = A^i\qth_i+\sum_jD^j_i\qth_j^2+\dots,
\]
or equivalently
\[
\qth_i = \frac{1}{A^i}\qs_i-\sum_j\frac{D^j_i}{A^iA^j}\qs_j^2+\dots,
\]
such that $\diff{}{\tilde \qt_0}$ and $\diff{}{\tilde\qt_1}$ are turned into vector fields of $D$-type in terms of the new odd variables $\qs_i$. Then to compute the central invariants, we only need to compute the coefficients $Z^i_a$ of $\qs_i^3$ in the $\deg_x = 3$ components of $\diff{u^i}{\tilde \qt_a}$ for $a = 0,1$. By a direct computation, it is easy to obtain
\begin{align*}
Z^i_0&= -\frac{f^iD^i_i}{(A^i)^3}-\frac{f^i B^i_i}{(A^i)^5},\\
Z^i_1 &= -\frac{u^if^iD^i_i}{(A^i)^3}-\frac{u^if^i B^i_i}{(A^i)^5}+3c_i\frac{(f^i)^2}{(A^i)^4}.
\end{align*}
By using the definition \eqref{BA},  we see that the central invariants $\tilde c_i$ of the bihamiltonian structure of the flow $\diff{}{\tilde t}$ read
\[
\tilde c_i = \frac{1}{3(\tilde f^i)^2}(Z^i_1-u^iZ^i_0) = c_i.
\]
Thus we complete the proof of Theorem \ref{AE}.

\section{Applications and examples}\label{ex}
\subsection{Classification of bihamiltonian integrable hierarchies by linear reciprocal transformations}
In this subsection, we classify, under the actions of both Miura-type transformations and linear reciprocal transformations, bihamiltonian integrable  hierarchies whose bihamiltonian structures admit semisimple hydrodynamic limits.

It follows from the theory of bihamiltonian cohomology that the Miura-equivalent classes of bihamiltonian structures are parameterized by their leading terms and central invariants, and leading terms are in one-to-one correspondence with flat pencils of metrics of the base manifolds. Therefore for two bihamiltonian integrable hierarchies, if the leading terms of their bihamiltonian structures, represented in their canonical coordinates,  coincide and their bihamiltonian structures share the same central invariants, then each flow of one bihamiltonian integrable hierarchy is a symmetry of the other bihamiltonian integrable hierarchy.

Next we consider the actions of both Miura-type transformations and linear reciprocal transformations on bihamiltonian integrable hierarchies. Since linear reciprocal transformations preserve central invariants, we conclude that two bihamiltonian integrable hierarchies are equivalent under actions of Miura-type transformations and linear reciprocal transformations if and only if the central invariants of their bihamiltonian structures coincide and their leading terms can be related by linear reciprocal transformations in their canonical coordinates. So we only need to consider the classification of bihamiltonian hierarchies of hydrodynamic type. 

Assume that we have two $n$-component semisimple bihamiltonian structures of hydrodynamic type $(P_0,P_1)$ and $(Q_0,Q_1)$ which, in terms of their canonical coordinates, are given by two metrics $\mathrm{diag}(f^1,\dots,f^n)$ and $\mathrm{diag}(g^1,\dots,g^n)$. As we have shown in Theorem \eqref{BR}, the necessary condition for the existence of a linear reciprocal transformation relating them is that these two metrics share a same set of rotation coefficients. Let us show that this is also a sufficient condition. Indeed, assume that their rotation coefficients coincide, i.e.,
\begin{equation}
\label{CA}
\qg_{ij}(u) = \frac{1}{2\sqrt{f_if_j}}\qp_jf_i = \frac{1}{2\sqrt{g_ig_j}}\qp_jg_i,\quad i\neq j,
\end{equation}
here as before we denote $f_i = (f^i)^{-1}$ and $g_i = (g^i)^{-1}$. If we can show that the flow
\[
\diff{u^i}{t} = \sqrt{\frac{f^i}{g^i}}u^{i,1}
\]
is a bihamiltonian vector field of $(P_0,P_1)$, then it follows from Theorem \ref{BR} that $(P_0,P_1)$ is transformed to $(Q_0,Q_1)$ after performing the linear reciprocal transformation with respect to the flow $\diff{}{t}$. According to the equation \eqref{BE}, we only need to check that
\[
\qp_j \sqrt{\frac{f^i}{g^i}} = \sqrt{\frac{f_j}{f_i}}\qg_{ij}\left(\sqrt{\frac{f^j}{g^j}}-\sqrt{\frac{f^i}{g^i}}\right),\quad i\neq j,
\]
which is easy to verify by using equations \eqref{CA}. Therefore we proved the following theorem.
\begin{Th}
The equivalent classes of $n$-component bihamiltonian integrable hierarchies under the actions of Miura-type transformations and linear reciprocal transformations are parameterized by the data $(\qg_{ij};c_1(u^1),\dots,c_n(u^n))$, where $\qg_{ij}$ is a solution of the system \eqref{BB}-\eqref{BD} and $c^1,\dots,c^n$ are $n$ functions of a single variable. 
\end{Th}

\begin{Ex}
Due to the dimensional reason, under the actions of Miura-type transformations and linear reciprocal transformations, equivalent classes of scalar bihamiltonian integrable hierarchies whose bihamiltonian structures possess semisimple hydrodynamic limits  are parameterized by functions $c_1(u^1)$ of a single variable. 
\end{Ex}

\begin{Ex}
Consider the Volterra hierarchy given by the following Lax equations:
\begin{equation}
\label{CL}
-4^k\qe\diff{L}{t_k} = [(L^{2k})_+,L], \quad L = \Ql+e^u\Ql^{-1},\quad k\geq 1,
\end{equation}
where $\Ql = \exp(\qe\qp_x)$ is the shift operator along the spatial variable. For example, we have
\begin{equation}
\label{CB}
\diff{u}{t_1} = -\frac{1}{4\qe}\left(\Ql-\Ql^{-1}\right)e^u.
\end{equation}
This integrable hierarchy admits a bihamiltonian structure $(P_0,P_1)$ with central invariant $-\frac{1}{12}$ and  Hamiltonian operators  \cite{adler1978trace,dubrovin2016hodge,faddeev1987hamiltonian}
\[
\mathcal P_0 = \frac{1}{8\qe}\left((\Ql+1)e^u(\Ql+1)-(\Ql^{-1}+1)e^u(\Ql^{-1}+1)\right),\quad \mathcal P_1 =  \frac{1}{2\qe}\left(\Ql-\Ql^{-1}\right).
\]
The leading terms of these operators are
\[
\mathcal P_0 ^{[0]} = e^u\qp_x+\frac12e^uu_x,\quad \mathcal P_1 ^{[0]} = \qp_x.
\]
In terms of the canonical coordinate $v = e^{-u}$ the above bihamiltonian structure of hydrodynamic type can be represented as
\[\tilde{\mathcal P}_0^{[0]} = v\qp_x+\frac12 v_x,\quad \tilde{\mathcal P}_1^{[0]} = v^2\qp_x+vv_x.
\]

Let us show that the first flow \eqref{CB} of the Volterra hierarchy is transformed to the first negative flow of the KdV hierarchy by a linear reciprocal transformation. Firstly, we need to adjust the central invariant of the bihamiltonian structure of the Volterra hierarchy to $\frac{1}{24}$ by rescaling the dispersion parameter as follows:
\[
\qe\mapsto \frac{\sqrt{-1}}{\sqrt 2}\qe.
\]
After this substitution, the flow \eqref{CB} can be represented in terms of $v$ by
\begin{align}
\label{CQ}
\diff{v}{t_1} =&\, -\frac{v_x}{2v}+\qe^2\left(\frac{v^{(3)}}{24v}-\frac{v_xv_{xx}}{4v^2}+\frac{v_x^3}{4v^3}\right)+\qe^4\left(-\frac{v^{(5)}}{960v}+\frac{v_xv^{(4)}}{96v^2}\right.\\
\notag
&\left.+\frac{v_{xx}v^{(3)}}{48v^2}-\frac{v_x^2v^{(3)}}{16v^3}-\frac{3v_xv_{xx}^2}{32v^3}+\frac{v_x^3v_{xx}}{4v^4}-\frac{v_x^5}{8v^5}\right)+O(\qe^6),
\end{align}
here and henceforth we will use $v^{(s)}$ to denote $\qp_x^sv$.
Note that for any smooth function $f(v)$, the flow
\[
\diff{v}{s^{[0]}} = f(v)v_x
\]
is a bihamiltonian vector field of $(P_0^{[0]},P_1^{[0]})$, therefore it determines a unique bihamiltonian vector field of $(P_0,P_1)$  which is a symmetry of the flow \eqref{CQ}. In particular, we have the following symmetry:
\begin{align}
\label{CR}
\diff{v}{s}=&\,\sqrt v v_x+\qe^2\sqrt v\left(\frac{v^{(3)}}{24}-\frac{v_xv_{xx}}{8v}+\frac{5v_x^3}{64v^2}\right)+\qe^4\sqrt v\left(\frac{v^{(5)}}{480}-\frac{v_xv^{(4)}}{96v}\right.\\
\notag
&\left.-\frac{v_{xx}v^{(3)}}{48v}+\frac{21v_x^2v^{(3)}}{512v^2}+\frac{v_xv_{xx}^2}{16v^2}-\frac{63v_x^3v_{xx}}{512v^3}+\frac{399v_x^5}{8192v^4}\right)+O(\qe^6).
\end{align}
Let us show that, after performing the linear reciprocal transformation with respect to the flow \eqref{CR} to the flow \eqref{CQ}, one obtains the first negative flow of the KdV hierarchy. Indeed, according to Theorem \ref{BR}, the two leading terms $\tilde{\mathcal P}_0^{[0]}$ and $\tilde{\mathcal P}_1^{[0]}$ of the  bihamiltonian structure are transformed to
\[
\tilde{\mathcal P}_0^{[0]}\mapsto \qp_x,\quad \tilde{\mathcal P}_1^{[0]}\mapsto v\qp_x+\frac12v_x,
\]
which are exactly the leading terms of the  bihamiltonian structure of the KdV hierarchy. Therefore we conclude that the flow \eqref{CQ} is transformed to a symmetry of the KdV hierarchy. By a straightforward computation, we see that the flow obtained from \eqref{CR} by the
 linear reciprocal transformation reads
\begin{align}
\label{CS}
\diff{v}{\tilde s}=&\,\frac{v_x}{\sqrt v} +\qe^2\frac{1}{\sqrt v}\left(-\frac{v^{(3)}}{24v}+\frac{5v_xv_{xx}}{16v^2}-\frac{23v_x^3}{64v^3}\right)+\qe^4\frac{1}{\sqrt v}\left(\frac{v^{(5)}}{320v^2}-\frac{3v_xv^{(4)}}{64v^3}\right.\\
\notag
&\left.-\frac{11v_{xx}v^{(3)}}{128v^3}+\frac{557v_x^2v^{(3)}}{1536v^4}+\frac{395v_xv_{xx}^2}{768v^4}-\frac{1791v_x^3v_{xx}}{1024v^5}+\frac{8571v_x^5}{8192v^6}\right)+O(\qe^6).
\end{align}
Let us perform the linear reciprocal transformation to the flow \eqref{CQ}, in another word, we make the following substitution to the flow \eqref{CQ}:
\[
\frac{\qp^n v}{\qp x^n}\mapsto \frac{\qp^n v}{\qp \tilde s^n}.
\]
Then the flow \eqref{CQ} is transformed to
\begin{align}
\label{CT}
\diff{v}{t} =&\, -\frac{v_x}{2v\sqrt v}+\qe^2\frac{1}{\sqrt v}\left(\frac{v^{(3)}}{16v^2}-\frac{19v_xv_{xx}}{32v^3}+\frac{107v_x^3}{128v^4}\right)+\qe^4\frac{1}{\sqrt v}\left(-\frac{v^{(5)}}{128v^3}+\frac{9v_xv^{(4)}}{64v^4}\right.\\
\notag
&\left.+\frac{64v_{xx}v^{(3)}}{256v^4}-\frac{1293v_x^2v^{(3)}}{1024v^5}-\frac{909v_xv_{xx}^2}{512v^5}+\frac{14193v_x^3v_{xx}}{2048v^6}-\frac{76467v_x^5}{16384v^7}\right)+O(\qe^6).
\end{align}
After performing the Miura-type transformation
\begin{align}
\label{CZ}
v\mapsto w =&\, v+\qe^2\left(\frac{3v_{xx}}{16v}-\frac{9v_x^2}{32v^2}\right)+\qe^4\left(\frac{3v^{(4)}}{512v^2}-\frac{3v_xv^{(3)}}{64v^3}\right.\\
\notag
&\left.-\frac{v_{xx}^2}{32v^3}+\frac{51v_x^2v_{xx}}{256v^4}-\frac{135v_x^4}{1024v^5}\right)+O(\qe^6),
\end{align}
the flow \eqref{CT} is transformed to
\begin{align}
\label{CU}
\diff{w}{ t} =&\, -\frac{w_x}{2w\sqrt w}+\qe^2\frac{1}{\sqrt w}\left(\frac{w^{(3)}}{16w^2}-\frac{5w_xw_{xx}}{16w^3}+\frac{35w_x^3}{128w^4}\right)+\qe^4\frac{1}{\sqrt w}\left(-\frac{w^{(5)}}{128w^3}+\frac{21w_xw^{(4)}}{256w^4}\right.\\
\notag
&\left.+\frac{35w_{xx}w^{(3)}}{256w^4}-\frac{483w_x^2w^{(3)}}{1024w^5}-\frac{651w_xw_{xx}^2}{1024w^5}+\frac{231w_x^3w_{xx}}{128w^6}-\frac{15015w_x^5}{16384w^7}\right)+O(\qe^6).
\end{align}
Let us verify that the flow \eqref{CU} is exactly the first negative flow of the KdV hierarchy. Recall that the negative flows of the KdV hierarchy are constructed in \cite{verosky1991negative} with the first negative flow given by
\begin{equation}
\label{CV}
\diff{w}{t} = p_x,\quad w_x = \left(w\qp_x+\frac12 w_x+\frac{\qe^2}{8}\qp_x^3\right)p.
\end{equation}
Other negative flows are defined recursively by the recursion operator
\[
w+\frac12w_x\qp_x^{-1}+\frac{\qe^2}{8}\qp_x^2.
\]
We can solve the variable $p$ as a differential polynomial in $w$ from the second equation of \eqref{CV}. For example, when we take the dispersionless limit $\qe = 0$, we obtain
\[
w_x = wp_x+\frac12w_xp,
\]
which can be integrated to obtain
\[
2\sqrt w+ C = \sqrt w p,
\]
here $C$ is the integral constant and we take $C = 1$ for simplicity. Then we can recursively solve the dispersion part of $p$ to obtain
\begin{align*}
p =&\, 2+\frac{1}{\sqrt w}+\qe^2\frac{1}{\sqrt w}\left(\frac{w_{xx}}{16w^2}-\frac{w_x^2}{64 w^3}\right)+\qe^4\frac{1}{\sqrt w}\left(-\frac{w^{(4)}}{128w^3}+\frac{7w^{(3)}w_x}{128w^4}\right.\\
&\left.+\frac{21w_{xx}^2}{512 w^4}-\frac{231w_{xx}w_x^2}{1024w^5}+\frac{1155w_x^4}{8192w^6}\right)+O(\qe^6).
\end{align*}
After substituting the expression of $p$ into
\[\diff{w}{t} = p_x,\]
we obtain the flow \eqref{CU}.

By a further computation, we can prove that actually the  hierarchy \eqref{CL} is transformed to the negative KdV hierarchy by performing the linear reciprocal transformation given by \eqref{CR}. Since we have already seen that the bihamiltonian structure of the Volterra hierarchy is transformed to that of the KdV hierarchy after the linear reciprocal transformation together with the Miura-type transformation \eqref{CZ}, we only need to check the above-mentioned relation for the leading terms of the two integrable hierarchies. Firstly, it is easy to see that the Volterra hierarchy \eqref{CL} has leading terms
\[
\diff{u}{t_k} = -\frac{k}{4^k}\binom{2k}{k}e^{ku}u_x+O(\qe).
\]
In terms of the canonical coordinate $v = e^{-u}$ we have
\[
\diff{v}{t_k} = -\frac{k}{4^k}\binom{2k}{k}\frac{v_x}{v^k}+O(\qe).
\]
After performing the linear reciprocal transformation, i.e., after making the substitution
\[
v_x\mapsto \frac{v_x}{\sqrt v}+O(\qe),
\]
we obtain
\begin{equation}
\label{DA}
\diff{v}{t_k} = -\frac{k}{4^k}\binom{2k}{k}\frac{v_x}{v^k\sqrt v}+O(\qe),\quad k\geq 1.
\end{equation}
These flows \eqref{DA} satisfy the recursion relations
\[
\left(v+\frac12v_x\qp_x^{-1}+\frac{\qe^2}{8}\qp_x^2\right)\diff{v}{t_{k+1}} = \diff{v}{t_k}.
\]
When $k = 1$, the flow \eqref{DA} is just the flow \eqref{CU}, so we conclude that the flows \eqref{DA} form the negative KdV hierarchy.
\end{Ex}
\subsection{Reciprocal transformations of the Ablowitz--Ladik hierarchy}
\label{CF}
In this subsection we consider the linear reciprocal transformations of the Ablowitz--Ladik hierarchy. The tri-Hamiltonian structure is given in \cite{li2022tri,oevel1989mastersymmetries} and the central invariants of the  associated bihamiltonian structures are presented in \cite{li2022tri}. We consider here the bihamiltonian structure $(P_0,P_1)$ (which is labeled by $(P_1,P_2)$ in the paper \cite{li2022tri}) with central invariants $c_1 = c_2 = \frac{1}{24}$ and with the semisimple leading term $(P_0^{[0]},P_1^{[0]})$. The Hamiltonian operators of this leading term read
\begin{align*}
\mathcal P_0^{[0]} &= \begin{pmatrix}
-2u^2\qp_x-u^{2,1} & -u^2\qp_x-u^{2,1} \\[6pt]
-u^2\qp_x & 0 
\end{pmatrix},\\[8pt]
 \mathcal P_1^{[0]} &= \begin{pmatrix}
0 & u^1u^2\qp_x+u^1u^{2,1} \\[6pt]
u^1u^2\qp_x+u^2u^{1,1} & 2(u^2)^2\qp_x+u^2u^{2,1} 
\end{pmatrix}.
\end{align*}
It has the canonical coordinates
\[
\ql^1 = 2u^2-u^1-2\sqrt{u^2(u^2-u^1)},\quad \ql^2 = 2u^2-u^1+\sqrt{u^2(u^2-u^1)},
\]
and we see that the semisimple bihamiltonian structure $(P_0^{[0]},P_1^{[0]})$ corresponds to the flat pencil $\mathrm{diag}(f^1,f^2)$ and $\mathrm{diag}(\ql^1f^1,\ql^2f^2)$, with
\[
f^1 = \frac{2\ql^1(\sqrt{\ql^1}+\sqrt{\ql^2})}{\sqrt{\ql^1}-\sqrt{\ql^2}},\quad f^2 = \frac{2\ql^2(\sqrt{\ql^1}+\sqrt{\ql^2})}{\sqrt{\ql^2}-\sqrt{\ql^1}}.
\]
By a direct computation we see that $\qg_{12} = \qg_{21}$. 

Let us find a hydrodynamic bihamiltonian vector field  $\diff{}{t^{[0]}}$ of $(P_0^{[0]},P_1^{[0]})$ with
\[
\diff{\ql^i}{t^{[0]}} = A^i\ql^{i,1},\quad i=1,2,
\]
such that the metric $\mathrm{diag}(\tilde f^1,\tilde f^2)$ given by 
\[
\tilde f^i = \frac{f^i}{(A^i)^2},\quad i=1,2
\]
satisfies the condition
\begin{equation}
\label{CE}
\diff{\tilde f^i}{\ql^1}+\diff{\tilde f^i}{\ql^2} = 0,\quad i=1,2.
\end{equation}
This is called the exactness condition, and we will give a more detailed and general discussion on this condition in Sect.\,\ref{dz}. For the present case, we can explicitly write down the equations \eqref{BE} and \eqref{CE} as follows:
\begin{align*}
\diff{A^1}{\ql^1} &=\frac{\ql^1-\ql^2-\sqrt{\ql^1\ql^2}}{2\ql^1(\ql^1-\ql^2)}A^1 -\frac{\sqrt{\ql^1}}{2\sqrt{\ql^2}(\ql^2-\ql^1)}A^2,\\
\diff{A^1}{\ql^2} &=\frac{\sqrt{\ql^1}}{2\sqrt{\ql^2}(\ql^2-\ql^1)}(A^2-A^1),\\
\diff{A^2}{\ql^1} &=\frac{\sqrt{\ql^2}}{2\sqrt{\ql^1}(\ql^1-\ql^2)}(A^1-A^2),\\
\diff{A^2}{\ql^2} &=\frac{\ql^1-\ql^2+\sqrt{\ql^1\ql^2}}{2\ql^2(\ql^1-\ql^2)}A^2-\frac{\sqrt{\ql^2}}{2\sqrt{\ql^1}(\ql^1-\ql^2)}A^1.
\end{align*}
The general solution of this system of PDEs can be represented by
\begin{align*}
A^1 &= \sqrt{\ql^1}(\sqrt{\ql^1}+\sqrt{\ql^2})\left(\frac{\qa}{\ql^1-\ql^2}+\qb\right),\\
A^2 &= \sqrt{\ql^2}(\sqrt{\ql^1}+\sqrt{\ql^2})\left(\frac{\qa}{\ql^2-\ql^1}+\qb\right),
\end{align*}
where $\qa,\qb$ are two arbitrary constants. Therefore we obtain
\begin{align*}
\tilde f^1 = \frac{2}{(\ql^1-\ql^2)\left(\frac{\qa}{\ql^1-\ql^2}+\qb\right)^2},\quad
\tilde f^2 = \frac{2}{(\ql^2-\ql^1)\left(\frac{\qa}{\ql^2-\ql^1}+\qb\right)^2}.
\end{align*}

It is easy to see that, when $\qa = 0,\qb = \frac12$ we get the leading term of the bihamiltonian structure of the extended nonlinear Schr\"{o}dinger hierarchy \cite{carlet2004extended,shabat1972exact}, and when $\qa = 2,\qb = 0$ the corresponding semisimple bihamiltonian structure is the leading term of the bihamiltonian structure of the extended Toda hierarchy \cite{carlet2004extended}. More explicitly, let us rewrite the flow
\begin{align*}
\diff{\ql^1}{t^{[0]}}&= \sqrt{\ql^1}(\sqrt{\ql^1}+\sqrt{\ql^2})\left(\frac{\qa}{\ql^1-\ql^2}+\qb\right)\ql^{1,1},\\
\diff{\ql^2}{t^{[0]}}&=\sqrt{\ql^2}(\sqrt{\ql^1}+\sqrt{\ql^2})\left(\frac{\qa}{\ql^2-\ql^1}+\qb\right)\ql^{2,1}
\end{align*}
in terms of the original coordinate $(u^1,u^2)$. Firstly, for $\qa = 0,\qb = \frac12$ we have
\begin{equation}
\label{DT}
\diff{u^1}{t^{[0]}} =  u^1u^{2,1},\quad \diff{u^2}{t^{[0]}} = 2 u^2u^{2,1}- u^2u^{1,1}.
\end{equation}
This flow coincides with the leading term of the first positive flow of the Ablowitz--Ladik hierarchy, see \cite{brini2012local,brini2012integrable,li2022tri,suris2012problem}. In the notations of \cite{li2022tri}, this flow reads
\begin{align}
\label{CX}
\diff{P}{t_0} &= \frac{1}{\qe}P(Q(x+\qe)-Q(x)),\\\label{CY} \diff{Q}{t_0} &= \frac{1}{\qe}Q(Q(x+\qe)-Q(x-\qe)-P(x)+P(x-\qe)).
\end{align} 
After the identification
\[
P(x,\qe = 0) = u^1(x),\quad  Q(x,\qe = 0) = u^2(x),
\]
the leading terms of \eqref{CX} and \eqref{CY} yield the flow \eqref{DT}. This is also the flow $\diff{}{t^{2,0}}$ of the Principal Hierarchy of the generalized Frobenius manifold corresponding to the Ablowitz--Ladik hierarchy \cite{li2022tri}.

For $\qa = 2,\qb = 0$, a direct calculation yields
\[
\diff{u^1}{t^{[0]}} = \frac{u^1(u^{2,1}-u^{1,1})}{u^2-u^1},\quad \diff{u^2}{t^{[0]}} = \frac{2u^2u^{2,1}-u^1u^{2,1}-u^2u^{1,1}}{u^2-u^1}.
\]
After the change of variables
\[r = \log u^2,\quad w = u^2-u^1\]
we obtain
\begin{equation}
\label{CP}
\diff{r}{t} = r_x+\frac{w_x}{w},\quad \diff{w}{t} = w_x+e^rr_x,
\end{equation}
which is the flow $\diff{}{t^{1,0}}$ of the Principal Hierarchy of the generalized Frobenius manifold corresponding to the Ablowitz--Ladik hierarchy (\cite{brini2012integrable,liu2022generalized}). It is expected that this flow is the leading term of a flow in a certain extension of the Ablowitz--Ladik hierarchy.

By using Theorem \eqref{ham-reci}, we can relate the the Principal Hierarchies associated with the Frobenius manifolds $M_{Toda}$, $M_{NLS}$ and $M_{AL}$, which correspond respectively to the Toda hierarchy, the nonlinear Schr\"{o}dinger hierarchy and the Ablowitz--Ladik hierarchy. The precise relations are summarized in the following table.
\vskip 0.3truecm
\begin{center}
    \begin{tabular}{ | c | c | c |}
    \hline & &\\[-7pt]
     Flow Transformed & Flow used for L.R.T. & Transformation Result \\[5pt] \hline & &\\[-7pt]
     $\diff{}{t^{1,p}_{Toda}}$ & $\diff{}{t^{2,0}_{Toda}}$ &$\diff{}{t^{2,p}_{NLS}}$  \\[7pt] \hline  & &\\[-7pt]
      $\diff{}{t^{2,p}_{Toda}}$ & $\diff{}{t^{2,0}_{Toda}}$ &$\diff{}{t^{1,p}_{NLS}}$  \\[7pt] \hline & &\\[-7pt]
       $\diff{}{t^{1,p}_{NLS}}$ & $\diff{}{t^{2,0}_{NLS}}$ &$\diff{}{t^{2,p}_{Toda}}$  \\[7pt] \hline & &\\[-7pt]
        $\diff{}{t^{2,p}_{NLS}}$ & $\diff{}{t^{2,0}_{NLS}}$ &$\diff{}{t^{1,p}_{Toda}}$  \\[7pt] \hline & &\\[-7pt]
         $\diff{}{t^{1,p}_{AL}}$ & $\diff{}{t^{1,0}_{AL}}$ &$\diff{}{t^{1,p}_{Toda}}$  \\[7pt] \hline & &\\[-7pt]
          $\diff{}{t^{1,p}_{AL}}$ & $\diff{}{t^{2,0}_{AL}}$ &$\diff{}{t^{2,p}_{NLS}}$  \\ [7pt] \hline & &\\[-7pt]
          
          $\diff{}{t^{2,p}_{AL}}$ & $\diff{}{t^{1,0}_{AL}}$ &$\diff{}{t^{2,p}_{Toda}}$  \\ [7pt] \hline & &\\[-7pt]
           $\diff{}{t^{2,p}_{AL}}$ & $\diff{}{t^{2,0}_{AL}}$ &$\diff{}{t^{1,p}_{NLS}}$  \\ [7pt]\hline
    \end{tabular}
\end{center}
\vskip 0.3truecm
Here for each row, the flow presented in the first column is transformed to the flow in the third column by performing the linear reciprocal transformation given by the flow in the second column.

Let us verify the relation given in, for example, the penultimate row in the above table.  We need to show that  when we perform the linear reciprocal transformation with respect to the flow $\diff{}{t^{1,0}_{AL}}$ in the Principal Hierarchy of $M_{AL}$, which is just the flow \eqref{CP}, then the flows  $\diff{}{t^{2,p}_{AL}}$ in the Principal Hierarchy of $M_{AL}$ are transformed to  $\diff{}{t^{2,p}_{Toda}}$ in the Principal Hierarchy of $M_{Toda}$ for any $p\geq 0$. Denote by $h^{Toda}_{p+1}$ and $h^{AL}_{p+1}$ the Hamiltonian densities of the flows $\diff{}{t^{2,p}_{Toda}}$ and $\diff{}{t^{2,p}_{AL}}$ with respect their first Hamiltonian structures. By using Theorem \ref{ham-reci}, it suffices to show that
 \[
 \diff{h^{Toda}_{p+1}}{\ql^1} = \frac{2\ql^1}{\sqrt{\ql^1}-\sqrt{\ql^2}}\diff{h^{AL}_{p+1}}{\ql^1},\quad  \diff{h^{Toda}_{p+1}}{\ql^2} = \frac{2\ql^2}{\sqrt{\ql^2}-\sqrt{\ql^1}}\diff{h^{AL}_{p+1}}{\ql^2}.
 \]
Note that the gradients of the Hamiltonian densities of Principal Hierarchies are determined by certain recursion relations and quasi-homogeneous conditions \cite{dubrovin1996geometry,dubrovin2001normal}, therefore the above identities can be proved by a straightforward computation by verifying these recursion relations and homogeneous conditions. All other rows in the table can be similarly verified.

According to the result of the bihamiltonian cohomology \cite{DLZ-1}, the bihamiltonian structure $(P_0,P_1)$ of the Ablowitz-Ladik hierarchy determines uniquely the deformations of the flows $\diff{}{t^{\qa,p}_{AL}}$ of the Principal Hierarchy of $M_{AL}$. More explicitly, the deformations of the flows $\diff{}{t^{2,p}_{AL}}$ coincide with the positive flows of the Ablowitz-Ladik hierarchy \cite{li2022tri} and the deformations of the flows $\diff{}{t^{1,p}_{AL}}$ define a certain extension of the Ablowitz-Ladik hierarchy. Since the central invariants of the bihamiltonian structures of the Toda hierarchy, the nonlinear Schr\"{o}dinger hierarchy and the Ablowitz-Ladik hierarchy are all equal to $\frac{1}{24}$, we conclude that the transformation relations presented in the table also hold true for the deformed flows up to Miura-type transformations. For example, by performing a Miura-type transformation if necessary, the positive flows of the Ablowitz-Ladik hiearchy are transformed to the flows of the Toda hierarchy by the linear reciprocal transformation with respect to the first extended flow of the Ablowitz-Ladik hiearchy.

\subsection{Reciprocal transformations and Dubrovin--Zhang hierarchies}
\label{dz}
We can generalize the discussion given in Sect.\,\ref{CF} and ask the following question: when a bihamiltonian integrable hierarchy can be transformed to symmetries of some Dubrovin--Zhang hierarchy \cite{dubrovin2001normal} by performing Miura-type transformations and linear reciprocal transformations?

Recall that the Dubrovin--Zhang hierarchies can be described as tau-symmetric bihamiltonian deformations of the Principal Hierarchies of semisimple Frobenius manifolds with the property that the central invariants of the deformed bihamiltonian structures are all equal to $\frac{1}{24}$ \cite{liu2021linearization}. Therefore we see that flows of a bihamiltonian integrable hierarchy can be transformed to symmetries of a certain Dubrovin--Zhang hierarchy if and only if the central invariants of their bihamiltonian structure are all equal to $\frac{1}{24}$, and their leading terms can be transformed to symmetries of the Principal Hierarchy of a certain semisimple Frobenius manifold. So again it suffices to consider only the hydrodynamic cases.

We will not review the theory of Frobenius manifolds and  associated Principal Hierarchies in this paper, one may refer to \cite{dubrovin1992integrable,dubrovin1996geometry,dubrovin1999painleve,dubrovin2001normal} for details. Instead, we review some results given in \cite{dubrovin2018bihamiltonian}. It is proved in \cite{dubrovin2018bihamiltonian} that a flat exact irreducible semisimple bihamiltonian structure $(P_0,P_1)$ must be the bihamiltonian structure of the Principal Hierarchy of a certain semisimple Frobenius manifold, and conversely bihamiltonian structures of Principal Hierarchies of  semisimple Frobenius manifolds are always flat exact. More precisely, assume that, in terms of the canonical coordinate $(u^1,\dots,u^n)$, the flat pencil corresponding to $(P_0,P_1)$ is given by $\mathrm{diag}(f^1,\dots,f^n)$ and $\mathrm{diag}(u^1f^1,\dots,u^nf^n)$, then the flat exact condition means that
\[
\sum_j \diff{f^i}{u^j} = 0,\quad i = 1,\dots,n,
\]
and the rotation coefficients satisfy the condition $\qg_{ij} = \qg_{ji}$ for $i\neq j$. The irreducible condition means that for any disjoint nonempty partition $I,J$ of the set $\{1,\dots,n\}$ and any index $i\in I$, there exists $j\in J$ such that $\qg_{ij}\neq 0$. 

Now we consider a bihamiltonian integrable hierarchy of hydrodynamic type with semisimple bihamiltonian structure $(P_0,P_1)$. We assume that $(P_0,P_1)$ corresponds to the flat pencil given by $\mathrm{diag}(f^1,\dots,f^n)$ and $\mathrm{diag}(u^1f^1,\dots,u^nf^n)$. Since linear reciprocal transformations preserve the rotation coefficients, the necessary condition of the existence of a linear reciprocal transformation to transform the flows of this integrable hierarchy to symmetries of the Principal Hierarchy of a semisimple Frobenius manifold is $\qg_{ij} = \qg_{ji}$. In what follows, we prove that this is also a sufficient condition under the irreducibility assumption.

Assume that $\qg_{ij} = \qg_{ji}$. In this case, the equations \eqref{BB}--\eqref{BD} read
\begin{align}
\label{CG}
&\qp_k\qg_{ij} = \qg_{ik}\qg_{kj},\quad i,j,k\  \text{distinct},\\
&\sum_k\qp_k\qg_{ij} = 0,\quad i\neq j,\\
\label{CH}
&\sum_ku^k\qp_k\qg_{ij} = -\qg_{ij},\quad i\neq j.
\end{align} 
Our goal is then to find a bihamiltonian vector field 
\[
\diff{u^i}{t} = A^iu^{i,1}
\]
of $(P_0,P_1)$ such that
\begin{equation}
\label{CI}
\sum_j\qp_j\tilde f^i = 0,\quad i=1,\dots,n,
\end{equation}
where
\[
\tilde f^i = \frac{f^i}{(A^i)^2}.
\]
From the the equations \eqref{BE} and \eqref{CI}, we obtain the following system of PDEs for $A^i$:
\begin{align}
\label{CJ}
\qp_i A^i& = \frac{A^i}{2f^i}\qp_if^i-\sum_{j\neq i}\sqrt{\frac{f^i}{f^j}}\qg_{ij}A^j,\\
\label{CK}
\qp_jA^i &= \sqrt{\frac{f_j}{f_i}}\qg_{ij}(A^j-A^i),\quad i\neq j.
\end{align}
By a direct computation, we can show that these equations are compatible, so the solution space is of dimension $n$. Let us prove that if the rotation coefficients $\qg_{ij}$ are irreducible, then we can find a solution $A^1,\dots,A^n$ with $A^i\neq 0$ for each $i  =1,\dots,n$. The argument is similar to that of Lemma 3.3 in \cite{dubrovin2018bihamiltonian}. For a nonempty subset $S\subseteq\{1,\dots,n\}$, let us denote
\[
A_S:=\prod_{i\in S}A^i,
\]
then we will prove that if $A_{\{1,\dots,n\}} = 0$, then for any nonempty subset $S\subseteq\{1,\dots,n\}$, $A_S = 0$ and in particular each $A^i = 0$. Let us prove by induction on $|S|$. For $|S| = n$, this is just our assumption that  $A_{\{1,\dots,n\}} = 0$. Assume that $A_S = 0$ for any $|S|\geq m$ for some $m\leq n$. Now take a nonempty subset $T\subseteq\{1,\dots,n\}$ with $|T|=m-1$ and take $i\in T$. By the irreducibility  assumption, we can find $j\notin T$ such that $\qg_{ij}\neq 0$. Denote $S = T\cup\{j\}$, then it follows from the induction hypothesis that $A_S = 0$. By using the equations \eqref{CJ} and \eqref{CK} we obtain
\begin{align*}
0 =&\,\diff{A_S}{u^i} = \sum_{k\in S}A_{S-\{k\}}\diff{A^k}{u^i}\\ 
=&\,\sum_{k\in S,k\neq i}A_{S-\{k\}}\sqrt{\frac{f_i}{f_k}}\qg_{ij}(A^i-A^k)\\&+A_{S-\{i\}}\frac{A^i}{2f^i}\qp_if^i-A_{S-\{i\}}\sum_{k\neq i}\sqrt{\frac{f^i}{f^k}}\qg_{ik}A^k\\
=&\,\sum_{k\in S,k\neq i}A_{S-\{k\}}\sqrt{\frac{f_i}{f_k}}\qg_{ij}A^i-A_{S-\{i\}}\sum_{k\in S,k\neq i}\sqrt{\frac{f^i}{f^k}}\qg_{ik}A^k.
\end{align*}
Thus it follows from the induction hypothesis that
\[
A_TA_{S-\{i\}} = 0;\quad A_TA_{S-\{k\}} = 0,\quad k\neq j,
\]
and we arrive at
\[
0 = A_T\diff{A_S}{u^i} = A_TA_{S-\{j\}}\sqrt{\frac{f_i}{f_j}}\qg_{ij}A^i = A_T^2\sqrt{\frac{f_i}{f_j}}\qg_{ij}A^i.
\]
By multiplying $A_{T-\{i\}}$ on both sides of the above identity, we obtain $A_T^3\qg_{ij} = 0$ and it follows that $A_T = 0$.

We summarize the results obtained in this subsection in the following theorem.
\begin{Th}
\label{DG}
Assume that we have a bihamiltonian integrable hierarchy and the central invariants of its bihamiltonian structure are all equal to $\frac{1}{24}$. If the rotation coefficients $\qg_{ij}$ of the flat pencil associated with the  leading term of the bihamiltonian structure satisfy the conditions
\[
\qg_{ij} = \qg_{ji},\quad i\neq j,\]
and if these coefficients are irreducible, then flows of  this integrable hierarchy can be transformed to symmetries of a certain Dubrovin--Zhang hierarchy by a linear reciprocal transformation.
\end{Th}

In the remaining part of this subsection, we show that two Dubrovin-Zhang hierarchies can be related by a linear reciprocal transformation if and only if their corresponding Frobenius manifolds can be related by a Legendre transformation. 

We start by recalling  the definitions of Legendre transformations for Frobenius manifolds \cite{dubrovin1996geometry,dubrovin2001normal}. Let $M$ be a Frobenius manifold with flat coordinates $v^1,\dots,v^n$. Let us denote by $\eta_{\qa\qb}$ the flat metric on $TM$ and by $c^\qa_{\qb\qg}(v)$ the structure constants of the Frobenius algebra on $T_vM$, i.e., on $TM$ we have 
\[
\langle\qp_\qa,\qp_\qb\rangle = \eta_{\qa\qb},\quad \qp_\qa\cdot\qp_\qb = c_{\qa\qb}^\qg\qp_\qg.
\]
 Choose a flat invertible vector field $b = b^\qa\qp_\qa$, then a Legendre transformation can be defined by 
\[
v^\qa\mapsto \hat v^\qa = b^\qb\eta^{\qa\qg}\qp_\qg\qp_\qb F,
\]
where $F$ is the potential of the Frobenius manifold $M$. Note that we have the relation
\[
\qp_\qa = b\cdot\hat\qp_\qa,
\]
here $\hat\qp_\qa = \diff{}{\hat v^\qa}$. It is proved in \cite{dubrovin1996geometry} that there exists a function $\hat F(\hat v)$ such that
\[
\hat\qp_\qa\hat\qp_\qb\hat F = \qp_\qa\qp_\qb F
\]
and that $\hat F$ defines a new Frobenius manifold structure on  $M$ with a new flat metric on $\hat M$ given by
\begin{equation}
\label{DF}
\langle X ,Y\rangle_b:=\langle b\cdot X ,b\cdot Y\rangle,\quad X,Y\in\Qg(TM),
\end{equation}
and the multiplication on $TM$ is not changed. Note that the metric $\langle - ,-\rangle_b$ has flat coordinates $\hat v^1,\dots,\hat v^n$ and
\[
\langle \hat\qp_\qa ,\hat\qp_\qb\rangle_b = \eta_{\qa\qb}.
\]
In the semisimple case, it is proved that Legendre transformations preserve the rotation coefficients \cite{Strachan2017}, therefore it follows from Theorem \ref{DG} that if two Frobenius manifolds are related by a Legendre transformation, then their corresponding Dubrovin-Zhang hierarchy is related by a linear reciprocal transformation. 

Now let us assume that the Principal Hierarchy of a semisimple Frobenius manifold $M$ can be transformed to symmetries of the Principal Hierarchy of another semisimple Frobenius manifold $M'$. Let us use $(u^1,\dots,u^n)$ to denote the canonical coordinates of $M$, then we know that the metric $\langle-,-\rangle$ is of diagonal form and we denote
\begin{equation}
\label{DK}
\langle du^i,du^j\rangle = \frac{\qd_{i,j}}{\psi_{i1}^2}.
\end{equation}
We define functions $\psi_{i\qa}$ by
\[
\diff{v^\qa}{u^i} = \psi_{i1}\psi_{i}^\qa,\quad \psi_{i}^\qa = \eta^{\qa\qb}\psi_{i\qb}.
\]
The following identities are useful \cite{dubrovin1996geometry,dubrovin2001normal}: 
\begin{align}
\label{DH}
c^\qa_{\qb\qg} &=\sum_i \frac{\psi_i^\qa\psi_{i\qb}\psi_{i\qg}}{\psi_{i1}},\\
\label{DI}
\diff{v^\qa}{u^i} &= \psi_{i1}\psi_{i}^\qa,\quad \diff{u^i}{v^\qa} = \frac{\psi_{i\qa}}{\psi_{i1}},\\
\label{DJ}
\diff{\psi_{i\qa}}{u^j} &= \qg_{ij}\psi_{j\qa},\quad i\neq j,\quad \sum_j\diff{\psi_{i\qa}}{u^j}  = 0.
\end{align}
If the Principal Hierarchy of $M$ can be transformed to symmetries of the Principal Hierarchy of $M'$ by performing the linear reciprocal transformation with respect to the flow
\begin{equation}
\label{DU}
\diff{u^i}{s} = A^iu^{i,1},
\end{equation}
then in this case the equations \eqref{CJ}, \eqref{CK} satisfied by $A^i$ read
\begin{align*}
&\qp_j A^i = \frac{\psi_{j1}}{\psi_{i1}}\qg_{ij}(A^j-A^i),\quad i\neq j,\\
&\sum_j\qp_jA^i = 0 .
\end{align*}
One can show that the general solution of the above PDEs have the from
\[
A^i = \frac{b^\qa\psi_{i\qa}}{\psi_{i1}},
\]
where $b^\qa$ are certain constants. So we assume that the functions $A^i$ in \eqref{DU} are given by the formulae above. It follows from Theorem \eqref{BR} that  after performing linear reciprocal transformation with respect to the flow
\[
\diff{u^i}{s} = \frac{b^\qa\psi_{i\qa}}{\psi_{i1}} u^{i,1},
\]
the metric \eqref{DK} is transformed to
\[
\langle-,-\rangle\mapsto \langle-,-\rangle'_b,
\]
where 
\begin{equation}
\label{DL}
\langle du^i,du^j\rangle'_b = \frac{\qd_{i,j}}{b^\qa b^\qb \psi_{i\qa}\psi_{i\qb}}.
\end{equation}
Let us check that this metric is exactly the metric defined by \eqref{DF} and hence we can choose $\hat v^1,\dots,\hat v^n$ as flat coordinates of $M'$ and this defines a Legendre transformation from $M$ to $M'$. For this purpose, we only need to compute the metric \eqref{DF} in terms of the canonical coordinate and prove that it coincides with \eqref{DL}. Firstly, we see that
\[
\langle \qp_\qa,\qp_\qb\rangle_b = \langle b\cdot\qp_\qa,b\cdot\qp_\qb\rangle=b^\mu b^\qg\eta_{\ql\qd}c^\ql_{\qa\mu}c^\qd_{\qb\qg}.
\]
Thus by using the identities \eqref{DH} and \eqref{DI} we arrive at
\[
\left\langle \diff{}{u^i},\diff{}{u^j}\right\rangle_b = \psi_{i1}\psi_{j1}\psi_i^\qa\psi_j^\qb\langle \qp_\qa,\qp_\qb\rangle_b = \qd_{i,j}{b^\qa b^\qb \psi_{i\qa}\psi_{i\qb}},
\]
here we also used the identities 
\[
\psi_{i\qa}\psi_j^\qa = \qd_{i,j},
\]
which  can be obtained from \eqref{DI}. Therefore we see that $\langle-,-\rangle'_b = \langle-,-\rangle_b$.

\begin{Ex}
It is proved in \cite{carlet2004extended} that the Toda hierarchy and the nonlinear Schr\"{o}dinger hierarchy can be related by a linear reciprocal transformation, therefore their associated Frobenius manifolds must be related by a Legendre transformation. Indeed, the associated Frobenius manifolds are given by the potentials
\[
F_{\textrm{Toda}} = \frac 12 v^2u+e^u,\quad F_{\textrm{NLS}} = \frac12v^2u+\frac12u^2\log u,
\]
and by choosing $b = \qp_u$, the potential $F_{\textrm{Toda}}$ can be transformed to $F_{\textrm{NLS}}$, see \cite{dubrovin1996geometry,dubrovin2001normal,dubrovin2018bihamiltonian} for details.
\end{Ex}

\section{Conclusion}\label{conc}
In this paper we study the actions of linear reciprocal transformations on bihamiltonian integrable hierarchies. We prove that these transformations preserve the bihamiltonian property, and the only invariants are the central invariants of the bihamiltonian structures and  rotation coefficients of the flat pencils associated with their leading terms.

One may notice that bihamiltonian cohomologies and their variational generalizations play important rules in the proof. So it is natural to ask if the methods presented in this paper work for an integrable hierarchy with only one Hamiltonian structure. More precisely, given flows $\diff{}{t}\in\derx^0_{\geq 1}$ and $\diff{}{\qt}\in\derx^1_{\geq 1}$ satisfying the condition
\[
\fk{\diff{}{t}}{\diff{}{\qt}}=
\fk{\diff{}{\qt}}{\diff{}{\qt}}=0,
\]
we ask whether one can find transformations of odd variables to turn them to vector fields of $D$-type if their leading terms are already of $D$-type. Unfortunately, this is not correct and one can easily write down some examples to justify this. So the method of transformations of odd variables cannot be used to answer the question whether linear reciprocal transformations preserve the Hamiltonian property for integrable hierarchies with only one Hamiltonian structure.

Another open problem is to study the general reciprocal transformations of integrable hierarchies. In \cite{ferapontov2003reciprocal}, the reciprocal transformations of Hamiltonian integrable hierarchies of  hydrodynamic type are studied and it is proved that in general  Hamiltonian structures are transformed to non-local ones.  We hope that we can develop a general theory for non-local Hamiltonian structures and study general reciprocal transformations by using the idea of super tau-cover introduced in \cite{liu2020super}.

The study of linear reciprocal transformations may also find its applications in the study of Gromov--Witten invariants. For example, in Sect.\,\ref{CF}, we relate the Ablowitz--Ladik hierarchy to the extended Toda hierarchy. It is well-known that the Gromov--Witten potential of $\mathbb P^1$ is a particular tau-function of the extended Toda hierarchy \cite{carlet2004extended,getzler2001toda}, and it is conjectured by Brini that the Gromov--Witten potential of local $\mathbb P^1$ is a tau-function of Ablowitz--Ladik hierarchy \cite{brini2012local}. Note that linear reciprocal transformations should preserve tau-functions, so the results in Sect.\,\ref{CF} probably suggest relations among Gromov--Witten invariants of $\mathbb P^1$ and those of local $\mathbb P^1$. 

\medskip

\noindent\textbf{Acknowledgments}
This work is supported by NSFC No.\,12171268.
\vskip 0.3truecm


\vskip 0.6truecm
\noindent Si-Qi Liu,

\noindent Department of Mathematical Sciences, Tsinghua University \\ 
Beijing 100084, P.R.~China\\
liusq@tsinghua.edu.cn
\medskip

\noindent Zhe Wang,

\noindent Department of Mathematical Sciences, Tsinghua University \\ 
Beijing 100084, P.R.~China\\
zhe-wang17@mails.tsinghua.edu.cn
\medskip

\noindent Youjin Zhang,

\noindent Department of Mathematical Sciences, Tsinghua University \\ 
Beijing 100084, P.R.~China\\
youjin@tsinghua.edu.cn

\end{document}